\documentclass[a4paper,11pt,reqno,english]{amsart}
\calclayout
\usepackage{geometry}
\usepackage{amssymb,amsmath,amsthm,amsfonts,mathtools,esint,dsfont,bbm}
\usepackage{appendix}
\usepackage{enumitem}
\usepackage{datetime}
\DeclareMathOperator{\tr}{Tr}
\usepackage{color}
	\definecolor{darkgreen}{rgb}{0,0.6,0}
	 \definecolor{darkred}{rgb}{0.8,0,0}
\usepackage{hyperref}

\theoremstyle{plain}
\newtheorem{theorem}{Theorem}[section]
\newtheorem{lemma}[theorem]{Lemma}
\numberwithin{equation}{section}

\newcommand\xqed[1]{%
	\leavevmode\unskip\penalty9999 \hbox{}\nobreak\hfill\quad\hbox{#1}%
}
\newcommand\remarkend{\xqed{$\triangle$}}
\makeatletter
	\def\@endtheorem{\remarkend\endtrivlist\@endpefalse }
\makeatother
\theoremstyle{remark}
\newtheorem{remark}[theorem]{Remark}
\newtheorem*{remark*}{Remark}
\makeatletter
	\def\@endtheorem{\endtrivlist\@endpefalse }
\makeatother

\newcommand\R{{\ensuremath {\mathbb R} }}
\newcommand\N{{\ensuremath {\mathbb N} }}
\newcommand\1{{\ensuremath {\mathds 1} }}
\newcommand{\gH}{\mathfrak{H}}
\newcommand{\cE}{\mathcal{E}}
\newcommand{\cG}{\mathcal{G}}
\newcommand{\cM}{\mathcal{M}}
\newcommand{\cQ}{\mathcal{Q}}
\newcommand{\cS}{\mathcal{S}}
\newcommand{\cT}{\mathcal{T}}
\newcommand{\cU}{\mathcal{U}}

\usepackage{xspace} 
\newcommand{\NLS}{\relax\ifmmode \mathrm{NLS}\else\textup{NLS}\xspace\fi}
\def\di{\mathop{}\!\mathrm{d}}

\def \dix{\di x}
\def\diy{\di y}
\def\diz{\di z}

\DeclareMathOperator{\sgn}{sign}

\newcommand{\normDISPLAY}[1]{{\left\vert\kern-0.25ex\left\vert #1 \right\vert\kern-0.25ex\right\vert}}
\newcommand{\normINLINE}[1]{{\vert\kern-0.25ex\vert #1 \vert\kern-0.25ex\vert{\vphantom{\normDISPLAY{1}}}}}
\newcommand{\norm}[1]{{\mathchoice{\normDISPLAY{#1}}{\normINLINE{#1}}{\normINLINE{#1}}{\normINLINE{#1}}}}
\newcommand{\normt}[1]{\left\vert\kern-0.25ex\left\vert\kern-0.25ex\left\vert #1 \right\vert\kern-0.25ex\right\vert\kern-0.25ex\right\vert}

\allowdisplaybreaks

\title[Stabilization of 2D BECs with repulsive, three\nobreakdash-body interactions]{Stabilization against collapse\\of 2D attractive Bose--Einstein condensates\\with repulsive, three\nobreakdash-body interactions}

\author[D.-T. Nguyen]{Dinh-Thi Nguyen}
	\address{(D.-T. Nguyen) Department of Mathematics, Uppsala University, Box 480, 751 06 Uppsala, Sweden}
	\email{\href{dinh_thi.nguyen@math.uu.se}{dinh\_thi.nguyen@math.uu.se}}

\author[J. Ricaud]{Julien Ricaud}
	\address{(J. Ricaud) CMAP, CNRS, {\'E}cole polytechnique, Institut Polytechnique de Paris, 91120 Pa\-lai\-seau, France}
	\email{\href{julien.ricaud@polytechnique.edu}{julien.ricaud@polytechnique.edu}}


\begin{document}

\begin{abstract}
	We consider a trapped Bose gas of~$N$ identical bosons in two dimensional space with both an attractive, two\nobreakdash-body, scaled interaction and a repulsive, three\nobreakdash-body, scaled interaction respectively of the form $-aN^{2\alpha-1} U(N^\alpha \cdot)$ and $bN^{4\beta-2} W(N^\beta \cdot, N^\beta \cdot))$, where $a,b,\alpha,\beta>0$ and $\int_{\mathbb R^2}U(x) {\mathop{}\mathrm{d}} x = 1 = \iint_{\mathbb R^{4}} W(x,y) {\mathop{}\mathrm{d}} x {\mathop{}\mathrm{d}} y$. We derive rigorously the cubic--quintic nonlinear Schr{\"o}\-dinger semiclassical theory as the mean-field limit of the model and we investigate the behavior of the system in the double-limit $a = a_N \to a_*$ and $b = b_N \searrow 0$. Moreover, we also consider the homogeneous problem where the trapping potential is removed.
\end{abstract}

\maketitle
\tableofcontents

\section{Introduction}

Bose--Einstein condensates (BECs)~\cite{Bose-24,Einstein-24,AndEnsMatWieCor-95,CorWie-02,Ketterle-02} are a form of matter made of bosonic particles that have been cooled down to temperatures near absolute zero, at which they condense into a single quantum state. In such context, a natural way to describe the interaction in a very general way is to consider an effective potential taking into account all the possible $k$\nobreakdash-body interactions, $k=2,\dots,N$, of the form $v_{\mathrm{kB}}(x_1-x_2,x_1-x_3,\dots,x_1-x_k)$ computed by integrating out the possible internal degrees of freedom of $k$ particles placed at $x_1$, $x_2$, ..., and $x_k$, in the ideal situation where all the others $N-k$ particles are infinitely far away:
\[
	\sum_{1\leq i < j \leq N} v_{\mathrm{2B}}(x_i-x_j) + \sum_{1\leq i < j < k \leq N} v_{\mathrm{3B}}(x_i-x_j,x_i-x_k) + \dots\,.
\]
If for most applications, limiting this effective interaction only to the two\nobreakdash-body term is enough for practical purposes, the three- and higher\nobreakdash-body terms of this many\nobreakdash-body expansion are actually not always negligible. For instance, it was shown in~\cite{MurBar-71} that $2\%$ of the binding energy of liquid~$\text{He}^4$ comes from three\nobreakdash-body interactions and in~\cite{MasBukSza-03} that it even amounts to $14\%$ for water, and that the two\nobreakdash-body approximation alone cannot explain certain of their physical properties~\cite{UjeVit-07,PieTaiSki-11}. Moreover, in the setting of an effective interaction composed of an attractive, two\nobreakdash-body potential and a repulsive, three\nobreakdash-body potential, setting in which the present work is interested, it was shown in~\cite{BisBla-15,Blakie-16} that the three\nobreakdash-body interaction can stabilize the condensate against a collapse due to the two\nobreakdash-body interaction and that the competing contributions can result in crystallization, making of it a good candidate setting to create super-solid states.

In the low-density regime, where the distance between particles is much larger than the interaction range, the pairwise interactions can be effectively treated as the dominant interactions as the three- and higher\nobreakdash-body interactions become negligible in comparison. This is due to the fact that the probability of having at least three particles simultaneously interacting becomes much smaller compared to two particles interacting at a time. As a result, the contributions from three- and higher\nobreakdash-body interactions can be neglected without significantly affecting the overall behavior of the system. In many cases, neglecting these effective interactions simplifies the analysis and calculations, leading the literature to focus primarily on the pairwise interactions between particles. This simplification is particularly useful in various fields, including condensed matter physics and atomic physics, where systems such as gases or liquids can be well-described by pairwise interaction potentials~\cite{LieYng-01,LieSei-02,LieSei-06}.

However, in the high-density regime, where particles are closely packed and the interparticle distances are comparable to or smaller than the range of interactions, the role of three- and higher\nobreakdash-body interactions can become more significant. In this regime, the effects of these interactions may need to be taken into account in order to accurately describe the system's behavior: in systems with strong three\nobreakdash-body interactions, the simultaneous interactions between three particles can lead to collective effects and correlations that cannot be captured by considering only pairwise interactions, and they can influence various properties of the system, such as its equation of state or its phase transitions.

In the past few years, the BECs with three\nobreakdash-body interactions received attention and many mathematical works dealt with three\nobreakdash-body interactions~\cite{ChePav-11,Chen-12,Yuan-15,Xie-15,CheHol-19,NamSal-20,LiYao-21,NamRicTri-22a,NamRicTri-22b,NamRicTri-23,NguRic-22-arxiv}. If most of the literature considers only a three\nobreakdash-body interaction, without a two\nobreakdash-body interaction, some works considered the combination of several few\nobreakdash-body interactions sharing the same scaling parameter, see~\cite{Xie-15,LiYao-21}. However, considering few\nobreakdash-body interactions with distinct scaling parameters can be of interest in some cases. For example (see~\cite{NguRic-22-arxiv}), in the one-dimensional case, this allows for the repulsive, two\nobreakdash-body interaction to compensate the (critical) attractive, three\nobreakdash-body interaction.

Finally, a well-known property of BEC for a system with (only) a two\nobreakdash-body interaction is that collapse occurs when the two\nobreakdash-body interaction is attractive and the number of atoms or the scattering length of the microscopic interaction exceeds a critical value~\cite{BraSacTolHul-95,BraSacHul-97,SacStoHul-98}.
One possible way to stabilize 2D focusing Bose gases against collapse is to take into account repulsive, three\nobreakdash-body interactions. These interactions can indeed counteract the attractive forces between the particles and prevent collapse from occurring. In particular, the repulsive, three\nobreakdash-body interactions can lead to the formation of stable droplets or clusters of particles, which can help to stabilize the system. This effect was first predicted in~\cite{EsrGreZhoLin-96} and studied, e.g., in~\cite{GamFreTom-99,AkhDasVag-99,GamFreTomCho-00}. Moreover, experimental studies have confirmed the stabilization of 2D Bose gases with repulsive, three\nobreakdash-body interactions against collapse~\cite{JosRic-97}.

\subsection{Model}
In this paper, we consider a 2D Bose gas with a scaled attractive, two\nobreakdash-body interaction and a scaled repulsive, three\nobreakdash-body interaction, trapped in a quasi 2D layer by means of a purely anharmonic potential. Namely, we consider a system of~$N\geq 3$ identical bosons in~$\R^2$, described by the nonrelativistic Hamiltonian
\begin{align}\label{hamiltonian}
	H_{a,b,N} := \sum_{i=1}^N\left(-\Delta_{x_i} + |x_i|^s\right) &- \frac{a}{N-1}\sum_{1\leq i<j\leq N} U_{N^\alpha}(x_i-x_j) \nonumber \\
	& + \frac{b}{(N-1)(N-2)} \sum_{1\leq i<j<k\leq N} W_{N^\beta}(x_i-x_j, x_i-x_k) \,,
\end{align}
acting on $L^2_{\mathrm{sym}}(\R^{2N})$. The parameter $s>0$ stands for the power of the anharmonic trap ($s=2$ being commonly used in laboratory experiments). Moreover, the parameters $a\geq0$ and $b\geq0$ are the strengths of the two- and three\nobreakdash-body interactions, respectively. The interaction terms $U_{N^\alpha}$ and $W_{N^\beta}$ are scaled through the parameters $\alpha,\beta>0$,
\begin{equation}\label{scaled-interaction}
	U_{N^\alpha}(x) := N^{2\alpha}U(N^\alpha x) \quad \text{and} \quad W_{N^\beta}(x,y) : = N^{4\beta}W(N^\beta x, N^\beta y) \,,
\end{equation}
in such a way that we may expect a well-defined semiclassical theory in the limit $N\to+\infty$. Finally, we assume on one hand that the two\nobreakdash-body interaction $-aU$ is attractive and even:
\begin{equation}\label{condition:potential-two-body}
	U(x)=U(-x) \geq 0, \quad \text{ with } \quad U \in L^1 \cap L^\infty(\R^2) \quad \text{ and } \quad \int_{\R^2}U(x) \dix = 1 \,.
\end{equation}
On the other hand, the three\nobreakdash-body interaction $bW$ has to be repulsive for matter stability and we assume
\begin{equation}\label{condition:potential-three-body}
	W(x,y) = W(y,x) \geq 0, \quad W \in L^1 \cap L^\infty(\R^2 \times \R^2), \quad \text{ and } \quad \iint_{\R^4} W(x,y) \dix \diy = 1\,,
\end{equation}
as well as the symmetry condition
\begin{equation}\label{condition:potential-three-body-symmetry}
	W(x-y, x-z)=W(y-x,y-z)=W(z-y,z-x)\,.
\end{equation}

We are interested in the large-$N$ behavior of the ground state energy of~$H_{a,b,N}$, given by
\begin{equation}\label{energy:quantum}
	E_{a,b,N}^{\mathrm{QM}} := N^{-1} \inf\left\{\langle\Psi_N |H_{a,b,N}| \Psi_N \rangle : \Psi_N \in L^2_{\mathrm{sym}}(\R^{2N}), \norm{\Psi_N}_{L^2}=1 \right\}
\end{equation}
and the corresponding ground states. The usual mean-field approximation suggests to restrict wave functions to the factorized ansatz of~$N$ particles
\begin{equation}\label{eq:BEC}
	\Psi_N(x_1,\ldots,x_N) \approx v^{\otimes N}(x_1,\ldots,x_N) := v(x_1)\ldots v(x_N) \,.
\end{equation}
Inserting in the energy functional the above uncorrelated state, with the normalization condition $\norm{v}_{L^2}=1$, we obtain the Hartree energy functional
\begin{align}\label{functional:hartree-inhomogeneous}
	\cE_{a, b,N}^{\mathrm{H}}(v) :=&{} \frac{\langle v^{\otimes N}, H_{a,b,N} v^{\otimes N} \rangle}{N} \nonumber \\
	=&{}
	\begin{multlined}[t]
		\int_{\R^2}\left( |\nabla v(x)|^2 + |x|^s |v(x)|^2 \right ) \dix - \frac{a}{2}\iint_{\R^4} U_{N^\alpha}(x-y)|v(x)|^2|v(y)|^2 \dix \diy \\
		+ \frac{b}{6}\iiint_{\R^6} W_{N^\beta}(x-y, x-z)|v(x)|^2|v(y)|^2|v(z)|^2 \dix \diy \diz \,.
	\end{multlined}
\end{align}
The corresponding Hartree ground state energy, given by
\begin{equation}\label{energy:hartree-inhomogeneous}
	E_{a,b,N}^{\mathrm{H}} := \inf\left\{\cE_{a, b,N}^{\mathrm{H}}(v):v\in H^1(\R^2), \norm{v}_{L^2}=1\right\},
\end{equation}
is thus an upper bound to the many\nobreakdash-body ground state energy: $E_{a,b,N}^{\mathrm{H}} \geq E_{a,b,N}^{\mathrm{QM}}$. The Hartree functional, which is commonly used to describe the mean-field approximation in many\nobreakdash-body systems, can be formally related to the \NLS functional under certain conditions. In our setting, as $N\to+\infty$, the potentials $U_{N^\alpha}(x-y)$ and $W_{N^\beta}(x-y, x-z)$ converge to the delta functions $\delta_{x=y}$ and $\delta_{x=y=z}$, respectively. The Hartree functional therefore formally boils down to the \NLS functional
\begin{equation}\label{functional:nls-inhomogeneous}
	\cE_{a,b}^{\NLS}(v) := \int_{\R^2}\left( |\nabla v(x)|^2 + |x|^s|v(x)|^2 - \frac{a}{2} |v(x)|^4 + \frac{b}{6} |v(x)|^6 \right) \dix \,,
\end{equation}
with the associated \NLS ground state energy
\begin{equation}\label{energy:nls-inhomogeneous}
	E_{a,b}^{\NLS} := \inf\left\{\cE^{\NLS}_{a,b}(v):v\in H^1(\R^2), \norm{v}_{L^2}=1\right\}.
\end{equation}

Note that the cubic term in~\eqref{functional:nls-inhomogeneous} is focusing ($a\geq0$), hence the quintic term has to be defocusing ($b\geq0$) in order to have a chance for the energy to be bounded below. This explains our assumption that the three\nobreakdash-body interaction is \emph{repulsive} ($bW \geq 0$) once the two\nobreakdash-body interaction has been assumed \emph{attractive} ($-aU \leq 0$).

The effective cubic--quintic functional~\eqref{functional:nls-inhomogeneous} is to the combinaison of a two- and a three\nobreakdash-body interactions what the cubic \NLS functional is to a purely two\nobreakdash-body interaction. The validity of the \NLS theory with a purely attractive, two\nobreakdash-body interaction potential ---$a>0=b$ in~\eqref{functional:nls-inhomogeneous}--- has been proved in seminal papers of Lewin, Nam, and Rougerie~\cite{Lewin-ICMP15,LewNamRou-16,LewNamRou-17,NamRou-20} (see also~\cite{CheHol-17b}), and the collapse and condensation of the many\nobreakdash-body ground states was studied in~\cite{LewNamRou-18-proc} (see also~\cite{GuoSei-14,Nguyen-20,GuoLuoYan-20,DinNguRou-23}). It is worth noting that, in the absence of three\nobreakdash-body interaction, the \NLS functional is unstable when the strength $a$ is above the critical mass $a_*$, which is the optimal constant in the Gagliardo--Nirenberg inequality
\begin{equation}\label{ineq:GN}
	\frac{a_*}{2} \norm{v}_{L^4}^4 \leq \norm{\nabla v}_{L^2}^2 \norm{v}_{L^2}^2,\quad \forall\, v\in H^1(\R^2) \,.
\end{equation}
Moreover, it is well-known that~\eqref{ineq:GN} has a positive radial symmetry optimizer $Q\in H^1(\R^2)$, which is the unique optimizer up to translations, multiplication by a complex factor, and scaling. Such optimizer solves the cubic \NLS equation (see~\cite{Weinstein-83})
\begin{equation}\label{eq:NLS}
	-\Delta Q + Q - Q^{3} = 0 \,.
\end{equation}
The critical value $a_*$ is then $a_* = \norm{Q}_{L^2}^2$ where $Q$ is the unique (up to translations) positive radial solution of~\eqref{eq:NLS}. Defining
\begin{equation}\label{Def_Q0}
	Q_0 := \norm{Q}_{L^2}^{-1} Q \,,
\end{equation}
it follows from~\eqref{ineq:GN} and~\eqref{eq:NLS} that
\begin{equation}\label{norms_Q_0}
	\norm{\nabla Q_0}_{L^2}^2 = \norm{Q_0}_{L^2}^2 = \frac{a_*}{2}\norm{Q_0}_{L^4}^4 = 1 \,.
\end{equation}

\subsection{Main results}
In this work, we derive rigorously the 2D cubic--quintic \NLS theory as the mean-field model of Bose gases. An interesting property that we establish is that the system is always stable due to the repulsive, three\nobreakdash-body interaction ($b>0$), but that the attractive, two\nobreakdash-body interaction can lead to the collapse of the system as $b$ tends to~$0$, depending on the (limit) value of~$a$. We study in particular the attractive BECs with repulsive, three\nobreakdash-body interactions in the \emph{critical} regime of collapse where $b$ tends to~$0$ while $a$ approaches the critical blow-up value $a_*$ (from above or below). This is obtained using the cubic--quintic \NLS functional (also known as Ginzburg--Gross--Pitaevskii).

In such collapse regime, we prove that the blowup profile of the \NLS ground states is given by the solution to the (critical) cubic \NLS in 2D. In the absence of three\nobreakdash-body interactions, this has been previously studied, e.g., in~\cite{GuoSei-14,LewNamRou-18-proc,GuoLuoYan-20,DinNguRou-23}.

In the many\nobreakdash-body theory, we consider the many-boson Hamiltonian~\eqref{hamiltonian}, interacting with rescaled potentials, and we prove the condensation and collapse of the many\nobreakdash-body ground states under the constraints $0 < \alpha < 1/8$ and $0 < \beta < 1/16$, using the interpolation Hartree theory. When the attractive, two\nobreakdash-body interaction is too negative ---namely, when $a \geq a_*$---, we need the extra assumption $\alpha<\beta$ in order for the three\nobreakdash-body interaction to compensate the two\nobreakdash-body interaction. Finally, in the collapse regime, we justify the validity of the \NLS theory with complete BEC at the point of blow-up.

It is worth noting that, for an attractive BEC with a repulsive, three\nobreakdash-body interaction, the system traps itself due to the attractive, two\nobreakdash-body interaction. Therefore, it is possible to remove the external potential and to consider the BECs in the full infinite space. In the 2D homogeneous case, an interesting phenomenon, which differs from the inhomogeneous case, occurs in the presence of a \emph{non-trivial}, repulsive, three\nobreakdash-body interaction: the attractive, two\nobreakdash-body interaction has to be sufficiently negative ($a>a_*$) in order to obtain the existence of \NLS ground states. The condensation in the mean-field limit and the collapse phenomenon are then also observed in the Hartree theory (see Appendix~\ref{app:homogeneous}). In the many\nobreakdash-body theory, however, a true many\nobreakdash-body wave function does not exist due to the translation-invariance and the linearity of the many\nobreakdash-body system. Even though we could consider \emph{approximate} ground states, their condensation is not expected due to their superposition in the many\nobreakdash-body theory. It stays therefore an \emph{open problem} to establish the convergence of quantum energy (in the infinite space) to the homogeneous \NLS energy.

\medskip

We now state precisely our main results. 
In the first part of the paper, we consider the minimization problem~\eqref{energy:nls-inhomogeneous}. Our first result is a classification of the values of the parameters $a$ and $b$ for which there exist \NLS ground states.
\begin{theorem}[Existence of \NLS ground states]\label{thm:existence-nls}
	Let $a>0$, $b\in\R$, and $E_{a,b}^{\NLS}$ be given in~\eqref{energy:nls-inhomogeneous}.
	\begin{enumerate}[label=(\roman*)]
		\item\label{nls-GS-inexistence} If ($b=0$ and $a>a_*$) or $b<0$, then $E_{a,b}^{\NLS} = -\infty$.
		\item\label{nls-GS-critical} If ($b=0$ and $a=a_*$), then $E_{a_*, 0}^{\NLS} = 0$, but there are no ground states.
		\item\label{nls-GS-existence} If ($b=0$ and $a < a_*$) or $b>0$, then $E_{a,b}^{\NLS}$ has a ground state.
	\end{enumerate}
\end{theorem}

A preliminary remark is that a ground state of~$E_{a,b}^{\NLS}$, when it exists, can be chosen to be nonnegative. This follows from the fact that $\norm{ \nabla v }_{L^2} \geq \norm{ \nabla |v| }_{L^2}$ for any $v \in H^1(\R^2)$ (see, e.g.,~\cite[Theorem 7.8]{LieLos-01}) and hence
$\cE_{a,b}^{\NLS}(v) \geq \cE_{a,b}^{\NLS}(|v|)$. In the absence of three\nobreakdash-body interactions, i.e., $b=0$, the classification of \NLS ground states was studied in~\cite{GuoSei-14}. In the case $b>0$, the repulsive quintic term counteracts the attractive cubic term since
\begin{equation}\label{ineq:2-to-3}
	\norm{v}_{L^4}^4 \leq \norm{v}_{L^6}^3 \norm{v}_{L^2} \leq \varepsilon \norm{v}_{L^6}^6 + \frac{1}{4\varepsilon} \norm{v}_{L^2}^2,\quad \forall \varepsilon > 0 \,,
\end{equation}
by the H\"older inequality, leading to a stabilization of the \NLS system against collapse.

However, Theorem~\ref{thm:existence-nls}-\emph{\ref{nls-GS-critical}} suggests that \NLS ground states might still collapse in the limit $a=a_n \to a_*$ and $b=b_n \searrow 0$, in the sense that $\norm{\nabla v_n}_{L^2} \to+\infty$ where $\{v_n\}_n$ is a sequence of ground states of~$E^{\NLS}_{a_n, b_n}$. Our next result therefore concerns the collapse of the \NLS ground states. Defining
\begin{equation}\label{Def-Q-s}
	 \cQ_6 := \frac{1}{6} \norm{Q_0}_{L^6}^6 \quad \text{and} \quad \cQ_s := s \norm{ |\cdot|^{\frac{s}{2}} Q_0 }_{L^2}^2
\end{equation}
for $Q_0$ in~\eqref{Def_Q0}, the variational principle gives
\begin{equation}\label{nls-energy-upperbound}
	E_{a,b}^{\NLS} \leq \cE_{a,b}^{\NLS}\left(\ell^{-1} Q_0(\ell^{-1} \cdot)\right) = \left(1 - \frac{a}{a_*}\right) \ell^{-2} + b \cQ_6 \ell^{-4} + \cQ_s \frac{\ell^{s}}{s}
\end{equation}
for all $\ell > 0$. By optimizing the right hand side of~\eqref{nls-energy-upperbound} over $\ell > 0$, it gives an upper bound to the \NLS energy. Depending on the collapse regime, we can actually also determine the main order in the expansion of the \NLS energy, by establishing explicitly the blowup profile of the corresponding \NLS ground states. Those results are summarized in the following theorem.

\begin{theorem}[Collapse of the \NLS ground states]\label{thm:collapse-nls}
	Let $s>0$, $\zeta \geq 0$. Let $\cQ_s$ and $\cQ_6$ be as in~\eqref{Def-Q-s}. Let $\{a_n\}_n,\{b_n\}_n \subset [0, +\infty)$ satisfy $a_n \to a_*$ and $b_n \searrow 0$ as $n\to+\infty$, with $\max\{0, a_* - a_n\} + b_n > 0$, and assume that
	\begin{equation}\label{thm-collapse-nls-condition-ratio-sequences}
		\lim_{n\to+\infty} \frac{2(a_* - a_n)}{a_* \cQ_s} \left( \frac{4 \cQ_6 b_n}{\cQ_s} \right)^{ -\frac{s+2}{s+4} } = (1-\zeta) \zeta^{-\frac{s+2}{s+4}} \in (-\infty,+\infty]\,.
	\end{equation}	
	Finally, let $\{v_n\}_n$ be a sequence of (approximate) ground states of~$E_{a_n, b_n}^{\NLS}$ given by~\eqref{energy:nls-inhomogeneous}. Then,
	\begin{equation}\label{thm:collapse-nls-inhomogeneous-ground-state}
		\lim_{n\to+\infty} \ell_n v_n(\ell_n \cdot) = Q_0
	\end{equation}
	strongly in~$H^1(\R^2)$ for the whole sequence, where $Q_0$ is given by~\eqref{Def_Q0} and
	\begin{equation}\label{thm:collapse-nls-blowup-length}
		\ell_n =
		\left\{
		\begin{aligned}
			&\left( \frac{2( a_* - a_n )}{a_* \cQ_s (1-\zeta)} \right)^{\frac{1}{s+2}} &\text{if } \zeta \neq 1\,, \\
			&\left(\frac{4 \cQ_6 b_n}{ \cQ_s \zeta} \right)^{\frac{1}{s+4}} &\text{if } \zeta \neq 0 \,.
		\end{aligned}
		\right.
	\end{equation}
	Furthermore,
	\begin{equation}\label{blowup:nls-energy}
		E_{a_n, b_n}^{\NLS} = \left( \frac{1}{2} + \frac{1}{s} - \frac{\zeta}{4} + o(1)\right) \cQ_s \ell_n^s \,.
	\end{equation}
\end{theorem}

\begin{remark}\label{rem:collapse-nls-inhomogeneous}\leavevmode
	\begin{itemize}[leftmargin=*]
		\item A sequence $\{v_n\}_n$ of \emph{approximate} ground states of~$E_{a_n, b_n}^{\NLS}$ is a sequence belonging to the minimizing domain and satisfying, as $n\to+\infty$,
		\begin{equation}\label{nls:inhomogeneous-approximate-ground-state}
			E_{a_n, b_n}^{\NLS} \leq \cE_{a_n, b_n}^{\NLS}(v_n) \leq (1+o(1))E_{a_n, b_n}^{\NLS} \,.
		\end{equation}
		
		\item The hypothesis $\max\{0, a_* - a_n\} + b_n > 0$ is equivalent to ($b_n>0$ or ($b_n = 0$ and $a_n<a_*$)), since we assume $a_n, b_n\geq0$. This hypothesis is made in order to ensure the existence of \NLS ground states, thanks to Theorem~\ref{thm:existence-nls}. 
		
		\item In the case $a_n \to a_*$ (from below or above) as $n\to+\infty$ while $b_n=b>0$ is fixed, $E_{a_n, b}^{\NLS}$ converges to $E_{a_*, b}^{\NLS}$ and the corresponding sequence of ground states converges to a ground state of~$E_{a_*, b}^{\NLS}$. The situation is similar in the case of a fixed $a_n=a < a_*$ while $b_n \searrow 0$ as $n\to+\infty$. Finally, note that the asymptotic behavior of~$E_{a_n,0}^{\NLS}$ and its ground states, in the limit $a_n \nearrow a_*$ as $n\to+\infty$, was already studied in~\cite{GuoSei-14}. 
		
		\item Our result in Theorem~\ref{thm:collapse-nls} covers all possible cases where the speeds of convergence of $a_n \to a_*$ and of $b_n \searrow 0$, as $n\to+\infty$, compare between them in a ``proportionality'' fashion (including the ``infinite proportionality'' case $a_n = a_*$), but with the notable exception of the case where $a_n$ converges from above ($a_n \searrow a_*$) with a convergence rate slower in order of magnitude than the one of $b_n \searrow 0$, which corresponds to the excluded case where the limit in~\eqref{thm-collapse-nls-condition-ratio-sequences} is $-\infty$, i.e., $\zeta=+\infty$.
		
			Roughly speaking, this latter case is similar to considering $a_n=a > a_*$ fixed and $b_n \searrow 0$. Within this framework, the \NLS functional is unstable, by Theorem~\ref{thm:existence-nls}-\emph{\ref{nls-GS-inexistence}}, and the \NLS ground states still blows up, in the sense that the \NLS energy is unbounded from below. We refer to Appendix~\ref{app:homogeneous} for further discussions.
	\end{itemize}
\end{remark}

In the second part of the paper, we turn to the $N$-particle Hamiltonian~\eqref{hamiltonian} with attractive, two- and repulsive, three\nobreakdash-body interaction potentials of the form~\eqref{scaled-interaction}. In the mean-field regime, we verify the validity of the effective cubic--quintic \NLS functional~\eqref{functional:nls-inhomogeneous}. As usual, the convergence of ground states is formulated using $k$-particles reduced density matrices, defined for $0\leq k \leq N$ and any $\Psi_N \in\gH^N$ as the partial trace
\[
	\gamma_{\Psi_N}^{(k)} := \tr_{k+1\to N} | \Psi_N \rangle \langle \Psi_N |\,.
\]
Equivalently, $\gamma_{\Psi_N}^{(k)}$ is the trace class operator on $\gH^k$ with kernel
\[
	\gamma_{\Psi_N}^{(k)}(x_1,\ldots,x_k;y_1,\ldots,y_k) := \int_{\R^{2(N-k)}}\overline{\Psi_N(x_1,\ldots,x_k;Z)}{\Psi_N(y_1,\ldots,y_k;Z)} \di Z \,.
\]
One of the main features of the reduced density matrices is that we can write
\[
	\frac{\langle\Psi_N | H_{a,b,N} | \Psi_N \rangle}{N} = \tr\left[ h \gamma_{\Psi_N}^{(1)} \right] - \frac{a}{2}\tr\left[U_{N^\alpha} \gamma_N^{(2)} \right] + \frac{b}{6}\tr\left[W_{N^\beta}\gamma_N^{(3)}\right],
\]
where $h_x := -\Delta_x+ |x|^s$ is the one-particle operator. Moreover, the Bose--Einstein condensation~\eqref{eq:BEC} is characterized properly by
\[
	\lim_{N\to+\infty} \tr \left|\gamma_{\Psi_N}^{(k)} - |v^{\otimes k}\rangle \langle v^{\otimes k}|\right| = 0,\quad \forall\, k=1,2,\dots\,.
\]

Our main result on the many\nobreakdash-body problem is the following.
\begin{theorem}[Condensation and collapse of the many\nobreakdash-body ground states]\label{thm:many-body}
	Let $0<\alpha<1/8$ and $0 < \beta < 1/16$. Assume that $U$ and $W$ satisfy~\eqref{condition:potential-two-body}--\eqref{condition:potential-three-body-symmetry}.
	\begin{enumerate}[label=(\roman*)]
		\item Let $a,b > 0$ be fixed with $a < a_*$ or ($a\geq a_*$ and $\alpha<\beta$). Then,
			\begin{align}\label{cv:qm-nls-energy}
				\lim_{N\to+\infty} E_{a,b,N}^{\mathrm{QM}} = E_{a,b}^{\NLS} >-\infty \,.
			\end{align}
			Moreover, for any sequence of ground states $\{\Psi_N \}_N$ of $E_{a,b,N}^{\mathrm{QM}}$ given by~\eqref{energy:quantum}, there exists a Borel probability measure~$\nu$ supported on 
			\[
				\cM^{\NLS} :=\{\NLS \text{ ground states}\}
			\]
			such that, along a subsequence,
			\begin{equation}\label{cv:qm-nls-ground-state}
				\lim_{N\to+\infty}\tr \left| \gamma_{\Psi_N}^{(k)} - \int |v^{\otimes k} \rangle \langle v^{\otimes k}| \di\nu(v) \right| =0,\quad \forall\, k=1,2,\dots\,.
			\end{equation}
			If $\cM^{\NLS}$ has, up to a phase, a unique ground state~$v$, then
			\[
				\lim_{N\to+\infty}\tr \left| \gamma_{\Psi_N}^{(k)} - |v^{\otimes k} \rangle \langle v^{\otimes k}| \right| =0,\quad \forall\, k=1,2,\dots
			\]
			for the whole sequence.
		\item Let $Q_0$, $\cQ_6$, $\cQ_s$, $\zeta$, $\{a_N\}_N$, $\{b_N\}_N$, and $\{\ell_N \}_N$ be as in Theorem~\ref{thm:collapse-nls}. Assume further that $|x|U(x) \in L^1(\R^2)$, that $\zeta \neq 2(s+2)/s$, that $\alpha<\beta$ if $\zeta \geq 1$, and that $\ell_N \sim N^{-\eta}$ with
		\begin{equation}\label{collapse:speed-many-body}
			0 < \eta < \min\left\{ \frac{\alpha}{s+3}, \frac{1-8\alpha}{4s}, \beta, \frac{1-16\beta}{4s} \right\}.
		\end{equation}
		Then,
		\begin{equation}\label{blowup:qm-energy}
			E_{a_N, b_N, N}^{\mathrm{QM}} = E_{a_N, b_N}^{\NLS} (1+o(1)) = \left( \frac{1}{2} + \frac{1}{s} - \frac{\zeta}{4} + o(1)\right) \cQ_s \ell_N^s \,.
		\end{equation}
		Moreover, for $\Phi_N = \ell_N \Psi_N(\ell_N \cdot)$ where $\{\Psi_N \}_N$ is any sequence of ground states of~$E_{a_N, b_N, N}^{\mathrm{QM}}$, given by~\eqref{energy:quantum}, we have
		\begin{equation}\label{blowup:qm-ground-state}
			\lim_{N\to+\infty} \tr \left| \gamma_{\Phi_N}^{(k)} - | Q_0^{\otimes k} \rangle \langle Q_0^{\otimes k} | \right| = 0,\quad \forall\, k=1,2,\dots\,.
		\end{equation}
	\end{enumerate}
\end{theorem}

It is worth noting that the condition $\alpha <\beta$ is imposed only if the attractive, two\nobreakdash-body interaction is too negative, i.e., $a\geq a_*$. This is needed to ensure that the three\nobreakdash-body interaction \emph{compensates} the two\nobreakdash-body interaction. This can be seen at the level of the Hartree theory (see Theorem~\ref{thm:collapse-hartree}). The condition is unnecessary when the strength of the two\nobreakdash-body interaction is below the threshold $a_*$, since in this case it can be controlled directly by the kinetic term using~\eqref{ineq:GN}. Moreover, note that the assumption $\zeta \neq 2(s+2)/s$ is related to the fact that the main order term in the \NLS energy is trivial at this value of $\zeta$.

The mathematical method used here is the information-theoretic quantum de~Finetti theorem developed recently by Rougerie~\cite{Rougerie-20a}. This method was used in the study of 2D Bose gases with only an attractive, two\nobreakdash-body interactions~\cite{NamRou-20} (see also~\cite{LewNamRou-16,LewNamRou-17} for related results) in the dilute regime, by using the second moment estimate. However, such an estimate is not available for systems with three\nobreakdash-body interactions in dimensional two and higher (see~\cite{NguRic-22-arxiv} for discussions in the one-dimensional case). Nevertheless, the method in~\cite{NamRou-20,Rougerie-20a} can still be coupled with the first moment estimate in order to derive BECs.

Finally, note that, in the mean-field regime and for fixed strengths, we can even obtain~\eqref{cv:qm-nls-energy} and~\eqref{cv:qm-nls-ground-state} for the wider range $0<\alpha<1/4$ and $0<\beta<1/8$ by refined arguments (see Remark~\ref{rem:bec-refined}). However, these arguments cannot be used in the collapse regime to prove~\eqref{collapse:speed-many-body}--\eqref{blowup:qm-ground-state} for a wider range in $(\alpha,\beta)$. Finally, note that more advance tools will be required if one wants to obtain BECs in the entire high-density regime $0<\alpha<1/2$ and $0<\beta<1/4$, and even more in the dilute regimes, where $\alpha>1/2$ or $\beta>1/4$.

\medskip
The remainder of the paper is devoted to the detailed mathematical proofs of our results. Moreover, in Appendix~\ref{app:homogeneous}, we state some results in the homogeneous BECs and give sketches of proofs.

\medskip
\noindent\textbf{Acknowledgments.} D.-T.~Nguyen was supported by the Knut and Alice Wallenberg Foundation. J.~Ricaud acknowledges support from the Agence Nationale de la Recherche under Grant No. ANR-19-CE46-0007 (project ICCI).

\section{Proofs of the main results}
\subsection{Existence and collapse in the \NLS theory}
The goal of this section is to establish the existence of the ground states for the \NLS problem~\eqref{energy:nls-inhomogeneous} and investigate its blowup profile as well as the first order expansion of its energy.

\begin{proof}[Proof of Theorem~\ref{thm:existence-nls}]
	Recall that~\eqref{nls-energy-upperbound} reads
	\[
		E_{a,b}^{\NLS} \leq \left(1 - \frac{a}{a_*}\right)\ell^{-2} + b \cQ_6 \ell^{-4} + \frac{\cQ_s}{s} \ell^s,\quad \forall\, \ell>0 \,.
	\]
	
	For~\ref{nls-GS-inexistence}, i.e., $b<0$ or ($b=0$ and $a>a_*$), taking $\ell \to 0$ yields the claim $E_{a,b}^{\NLS} = -\infty$.
	
	For~\ref{nls-GS-critical}, i.e., $b=0$ and $a=a_*$, on one hand we obtain $E_{a_*, 0}^{\NLS} \leq 0$ the same way and, on the other hand, the converse inequality follows from~\eqref{ineq:GN}. Now, if a ground state~$u$ for $E_{a_*, 0}^{\NLS} = 0$ would exist, then~\eqref{ineq:GN} would imply $\norm{ |\cdot|^{s/2} u }_{L^2} = 0$, contradicting $\norm{u}_{L^2}=1$.
	
	The rest of the proof is dedicated to~\ref{nls-GS-existence}: the existence of \NLS ground states when $b>0$ or ($b=0$ and $a < a_*$). Recall that Theorem~\ref{thm:existence-nls} assumes $a>0$. Let $\{ v_n \}_n$ be a minimizing sequence for $E_{a,b}^{\NLS}$, i.e.,
	\[
		\lim_{n\to+\infty}\cE_{a,b}^{\NLS}(v_n) = E_{a,b}^{\NLS} \quad \text{ with } \quad \norm{v_n}_{L^2}=1 \,.
	\]
	The existence of \NLS ground states is obtained (see, e.g.,~\cite{GuoSei-14}) by claiming that $\{ \nabla v_n \}_n$ and $\{ |\cdot|^{s/2} v_n \}_n$ are uniformly bounded in~$L^2(\R^2)$. This follows immediately from~\eqref{ineq:GN} if $a < a_*$ and $b\geq0$. In the remaining case, where $a \geq a_*$ and $b>0$, the first inequality in~\eqref{ineq:2-to-3} yields
	\[
		\cE_{a,b}^{\NLS}(v_n) \geq \norm{\nabla v_n}_{L^2}^2 + \norm{ |\cdot|^{\frac{s}{2}} v_n }_{L^2}^2 - \frac{a}{2} \norm{v_n}_{L^4}^4 + \frac{b}{6} \norm{v_n}_{L^4}^8 \geq \norm{\nabla v_n}_{L^2}^2 + \norm{ |\cdot|^{\frac{s}{2}} v_n }_{L^2}^2 - C_{a,b} \,,
	\]
	implying the claim on $\{ \nabla v_n \}_n$ and $\{ |\cdot|^{s/2} v_n \}_n$ since $\cE_{a,b}^{\NLS}(v_n)$ is uniformly bounded.
\end{proof}

\begin{proof}[Proof of Theorem~\ref{thm:collapse-nls}]
	The proof for approximate ground states being the same as the one for true ground states, with very few changes, we only write the latter for brevity. Also for shortness, we define $A_n := 1 - a_n/a_*$.

	We start with the upper bound matching~\eqref{blowup:nls-energy}. We obtain it by inserting $\ell=\ell_n$, defined in Theorem~\ref{thm:collapse-nls}, into~\eqref{nls-energy-upperbound}, then rewriting the upper bound using~\eqref{thm:collapse-nls-blowup-length}:
	\begin{equation}\label{nls-upperbound}
		E_{a_n, b_n}^{\NLS} \leq \ell_n^s \left(A_n \ell_n^{-s-2} + \cQ_6 b_n \ell_n^{-s-4} + \frac{\cQ_s}{s} \right) = \left( \frac{1}{2} + \frac{1}{s} - \frac{\zeta}{4} + o(1)\right) \cQ_s \ell_n^s \,.
	\end{equation}
	Note that the computation of the equality is performed separately for $\zeta\neq0$ and $\zeta\neq1$, but both yield the same expansion.
	
	We prove the matching lower bound by proving the claimed convergence~\eqref{thm:collapse-nls-inhomogeneous-ground-state}. Let $\{v_n\}_n$ be a sequence of ground states of~$E_{a_n, b_n}^{\NLS}$ and $w_n := \ell_n v_n(\ell_n \cdot)$ for any $n$. Then, $\norm{w_n}_{L^2} = 1 = \norm{v_n}_{L^2}$ and
	\begin{equation}\label{nls-approximate}
		E_{a_n, b_n}^{\NLS} = \ell_n^s \left[ \ell_n^{-s-2} \left(\norm{ \nabla w_n }_{L^2}^2 - \frac{a_n}{2} \norm{w_n}_{L^4}^4 \right) + \frac{b_n}{6} \ell_n^{-s-4}\norm{w_n}_{L^6}^6 + \norm{ |\cdot|^{\frac{s}{2}} w_n }_{L^2}^2 \right].
	\end{equation}
	We distinguish two (overlapping) cases for~$\zeta \geq 0$: $\zeta\neq0$ and $\zeta\neq1$.
	
	\medskip
	\noindent\textbf{Case $\zeta\neq0$.}
	Applying~\eqref{ineq:GN} to the kinetic term in~\eqref{nls-approximate} and rewriting the lower bound obtained that way, using~\eqref{thm-collapse-nls-condition-ratio-sequences} and the definition of~$\ell_n$ depending on $b_n$ in~\eqref{thm:collapse-nls-blowup-length}, we obtain by~\eqref{nls-upperbound} that
	\begin{equation}\label{nls_control_nrj_L2_for_gamma_not_0}
		\frac{1}{2} + \frac{1}{s} - \frac{\zeta}{4} + o(1) \geq \frac{a_*}{2} \left(\frac{1-\zeta}{2} +o(1)\right) \norm{w_n}_{L^4}^4 + \frac{\zeta}{24 \cQ_6} \norm{w_n}_{L^6}^6 + \cQ_s^{-1} \norm{ |\cdot|^{\frac{s}{2}} w_n }_{L^2}^2 \,.
	\end{equation}
	Using moreover~\eqref{ineq:2-to-3} on the nonnegative $L^6$-norm term yields that $\{w_n\}_n$ is uniformly bounded in~$L^4(\R^2)$ and consequently, by~\eqref{nls_control_nrj_L2_for_gamma_not_0}, that $\{w_n\}_n$ and $\{ |\cdot|^{s/2} w_n \}_n$ are uniformly bounded respectively in~$L^6(\R^2)$ and~$L^2(\R^2)$.
	
	We now multiply~\eqref{nls-approximate} by $\ell_n^2$, neglect in it the nonnegative quintic and external terms, and use~\eqref{nls-upperbound}. On the one hand, using the $L^4$-boundedness of~$\{w_n\}_n$, it gives that $\{\norm{ \nabla w_n }_{L^2}\}_n$ is bounded uniformly as well and, consequently, that there exists $w\in H^1(\R^2)$ such that, up to a subsequence, $w_n \to w$ weakly in~$H^1(\R^2)$ and strongly in~$L^r(\R^2)$ for $2\leq r < \infty$. Hence, in particular, $\norm{ w }_{L^2}=1$. On the other hand, passing to the limit $n\to+\infty$ yields
	\[
		0 \geq \lim_{n\to+\infty} \left(\norm{ \nabla w_n }_{L^2}^2 - \frac{a_n}{2} \norm{w_n}_{L^4}^4 \right) \geq \norm{ \nabla w }_{L^2}^2 - \frac{a_*}{2} \norm{ w }_{L^4}^4 \geq 0 \,,
	\]
	where we have used Fatous' lemma and~\eqref{ineq:GN}.
	
	It follows from the above that $w_n \to w$ strongly in~$H^1(\R^2)$ and that $w$ is an optimizer of~\eqref{ineq:GN}. Hence, after a suitable rescaling, $w(x)=\sqrt{t}Q_0(\sqrt{t}x+x_0)$ for some $t>0$ and $x_0\in\R^2$. We now show that $w \equiv Q_0$ by proving that $t= 1$ and $x_0 = 0$. Indeed, taking the limit $n\to+\infty$ in~\eqref{nls_control_nrj_L2_for_gamma_not_0} then writing things in terms of $Q_0$ (instead of $w$) gives
	\begin{align}\label{nls-control_nrj-L4-for-gamma_not-0}
		\frac{1}{2} + \frac{1}{s} - \frac{\zeta}{4}
		&= t \frac{1-\zeta}{2} \frac{a_*}{2} \norm{Q_0}_{L^4}^4 + t^2 \frac{\zeta}{24 \cQ_6} \norm{Q_0}_{L^6}^6 + t^{-\frac{s}{2}} \cQ_s^{-1} \norm{ \left|\cdot-x_0\right|^{\frac{s}{2}} Q_0 }_{L^2}^2 \nonumber \\
		&\geq t \frac{1-\zeta}{2} + t^2 \frac{\zeta}{4} + \frac{t^{-\frac{s}{2}}}{s} =: g_{\zeta}(t)
	\end{align}
	where, for the inequality, we use the fact that $Q_0$ is \emph{strictly symmetric decreasing} and $|\cdot|^s$ is \emph{strictly symmetric increasing} ---see, e.g., \cite[Appendix~A]{NguRic-22-arxiv} for details---, as well as the value of the $L^4$-norm of~$Q_0$, given in~\eqref{norms_Q_0}, and the definition of~$\cQ_6$. Notice that 
	\begin{align*}
		g_{\zeta}': (0,+\infty) & \to (-\infty,+\infty)\\
		t & \mapsto \frac{1-\zeta}{2} + t\frac{\zeta}{2} - \frac{t^{-\frac{s}{2}-1}}{2}
	\end{align*}
	is strictly increasing and $g_{\zeta}'(1) = 0$. Hence, $g_{\zeta}$ attains its global minimum at $t=1$ with 
	\[
		g_{\zeta}(1) = \frac{1}{2} + \frac{1}{s} - \frac{\zeta}{4} \,.
	\]
	Consequently, equality holds in~\eqref{nls-control_nrj-L4-for-gamma_not-0}, which implies that $t=1$ and consequently ---see, e.g., \cite[Lemma~A.1]{NguRic-22-arxiv}--- $x_0 = 0$. Finally, the convergence~\eqref{thm:collapse-nls-inhomogeneous-ground-state} holds for the whole sequence due to the uniqueness of the limiting profile, see, e.g., \cite[Proof of Theorem~1.2]{NguRic-22-arxiv} for details.

	\medskip
	\noindent\textbf{Case $\zeta\neq 1$.}
	Applying~\eqref{ineq:GN} to the $L^4$-norm term in~\eqref{nls-approximate} and rewriting the lower bound obtained that way, using~\eqref{thm-collapse-nls-condition-ratio-sequences} and the definition of~$\ell_n$ depending on $A_n$ in~\eqref{thm:collapse-nls-blowup-length} ---it is well defined at least for $n$ large enough since, by~\eqref{thm-collapse-nls-condition-ratio-sequences}, $A_n/(1-\zeta) > 0$ for $n$ large enough---, we obtain by~\eqref{nls-upperbound} that
	\[
		\frac{1}{2} + \frac{1}{s} - \frac{\zeta}{4} + o(1) \geq \frac{1-\zeta}{2}\norm{ \nabla w_n }_{L^2}^2 + \left(\frac{\zeta}{24 \cQ_6} + o(1)\right)\norm{w_n}_{L^6}^6 + \cQ_s^{-1}\norm{ |\cdot|^{\frac{s}{2}} w_n }_{L^2}^2 \,.
	\]
	For $0\leq \zeta<1$ the above yields the uniformly boundedness of~$\{\nabla w_n\}_n$ and $\{ |\cdot|^{s/2} w_n \}_n$ in~$L^2(\R^2)$. This actually still holds true when $\zeta > 1$, by the same arguments as in the case $\zeta \ne 0$. For the rest of the proof, we omit the details since it strictly follows the end of the proof in the case $\zeta \neq 0$.	
\end{proof}

\subsection{Condensation and collapse in the Hartree theory}
While the quantum energy describes the exact energy of the system, taking into account the whole quantum interaction between the particles, the Hartree energy on its side arises from the mean-field approximation, where the interactions are treated as averaged or as effective fields experienced by each particle. The Hartree theory plays a role of an interpolation theory between the many\nobreakdash-body and the \NLS theories, and this role is important in the study of the systems of two- and three\nobreakdash-body interactions, as in~\cite{NguRic-22-arxiv}. Therefore, understanding microscopic phenomenon in the Hartree theory is the next step before turning to the many\nobreakdash-body problem. We have the following result.

\begin{theorem}[Condensation and collapse of the Hartree ground states]\label{thm:collapse-hartree}
	Let $\alpha,\beta > 0$ and let $U$ and $W$ satisfy~\eqref{condition:potential-two-body}--\eqref{condition:potential-three-body-symmetry}.
	\begin{enumerate}[label=(\roman*)]
		\item\label{thm:collapse-hartree_item1} Let $a,b > 0$ be fixed with $a < a_*$ or ($a\geq a_*$ and $\alpha<\beta$). Let $\{v_N\}_N$ be a sequence of (approximate) ground states of~$E_{a,b,N}^{\mathrm{H}}$ given by~\eqref{energy:hartree-inhomogeneous}. Then, there exists a \NLS ground state~$v$ of~\eqref{energy:nls-inhomogeneous} such that, along a subsequence,
		\begin{equation}\label{cv:hartree-nls-ground-state}
			\lim_{N\to+\infty}v_N = v
		\end{equation}
		strongly in~$H^1(\R^2)$. Furthermore,
		\begin{equation}\label{cv:hartree-nls-energy}
			\lim_{N\to+\infty}E_{a,b,N}^{\mathrm{H}} = E_{a,b}^{\NLS} \,.
		\end{equation}
		
		\item\label{thm:collapse-hartree_item2} Let $Q_0$, $\cQ_6$, $\cQ_s$, $\zeta$, $\{a_N\}_N$, $\{b_N\}_N$, and $\{\ell_N \}_N$ be as in Theorem~\ref{thm:collapse-nls}. Assume further that $|x|U(x) \in L^1(\R^2)$, that $\alpha<\beta$ if $\zeta \geq 1$, and that $\ell_N \sim N^{-\eta}$ with
		\begin{equation}\label{collapse:speed-hartree}
			0 < \eta < \min\left\{ \frac{\alpha}{s+3}, \beta \right\}.
		\end{equation}
		Then,
		\begin{equation}\label{blowup:hartree-inhomogeneous-energy}
			E_{a_N, b_N, N}^{\mathrm{H}} = E_{a_N, b_N}^{\NLS} + o\!\left(\ell_N^s\right) = \left( \frac{1}{s} + \frac{1}{2} - \frac{\zeta}{4} + o(1) \right) \cQ_s \ell_N^s \,.
		\end{equation}
		Moreover, for any sequence of (approximate) ground states $\{v_N\}_N$ of~$E_{a_N, b_N, N}^{\mathrm{H}}$ given by~\eqref{energy:hartree-inhomogeneous}, we have
		\begin{equation}\label{blowup:hartree-inhomogeneous-ground-state}
			\lim_{N\to+\infty} \ell_N v_N ( \ell_N \cdot ) = Q_0
		\end{equation}
		strongly in~$H^1(\R^2)$, for the whole sequence.
	\end{enumerate}
\end{theorem}

\begin{remark*}\leavevmode
	\begin{itemize}[leftmargin=*]
		\item Similarly to the \NLS case, a sequence $\{v_N\}_N$ of \emph{approximate} ground states of~$E_{a_N, b_N, N}^{\mathrm{H}}$ is a sequence belonging to the minimizing domain and satisfying
			\begin{equation}\label{hartree:inhomogeneous-approximate-ground-state}
				E_{a_N, b_N, N}^{\mathrm{H}} \leq \cE_{a_N, b_N, N}^{\mathrm{H}}(v_N) \leq E_{a_N, b_N, N}^{\mathrm{H}} \left( 1 + o(1)_{N\to+\infty} \right) \,.
			\end{equation}
		\item Here again, the repulsive potential compensates for the attractive one. However, while for the \NLS theory it could be seen directly due to~\eqref{ineq:2-to-3}, for the Hartree theory this fact is not observed easily. In particular, when the strength $a$ is too negative ($a \geq a_*$), the two\nobreakdash-body interaction is not controlled anymore by the kinetic term. In that case, we need the repulsive, three\nobreakdash-body interaction to compensate the the attractive, two\nobreakdash-body interaction and it is why we need to impose the condition $\alpha<\beta$ in that case.
		\item The condition $\ell_N \sim N^{-\eta}$ is, by definition of~$\ell_N$ in~\eqref{thm:collapse-nls-blowup-length}, also a condition on $a_* - a_N$ or on $b_N$, depending on the value of~$\zeta$. The technical assumption on the attractive, two\nobreakdash-body interaction, i.e., $|x|U(x) \in L^1(\R^2)$ is used to determine the convergence rate of the Hartree energy to the \NLS energy in the limit $N\to+\infty$. The condition~\eqref{collapse:speed-hartree} is then used to ensure that the Hartree and \NLS ground state problems are close in the collapse regime. Moreover, the extra condition $\alpha<\beta$ applies in particular to the case where $a_N \searrow a_*$. This corresponds to the case $\zeta \geq 1$ in Theorem~\ref{thm:collapse-nls} and ensures that the two\nobreakdash-body interaction, which is too negative ---namely, above the threshold $a_*$---, is still controlled by the three\nobreakdash-body interaction.
		\item The compactness of the sequence of Hartree ground states is obtained by refined arguments using the singularity of the repulsive, three\nobreakdash-body interaction.
	\end{itemize}
\end{remark*}

In order to prove Theorem~\ref{thm:collapse-hartree}, we need the following lemma, which is about the (rate of) convergence of the two- and three\nobreakdash-body interactions.

\begin{lemma}\label{lem:hartree-nls-2-3-body}
	Assume that $U$ and $W$ satisfy~\eqref{condition:potential-two-body}--\eqref{condition:potential-three-body-symmetry}. Then, for any $v\in H^1(\R^2)$, we have
	\begin{align}\label{cv-ineq:hartree-nls-2body}
		0 \leq \norm{v}_{L^4}^4 - \iint_{\R^4} U_{N^\alpha}(x-y) |v(x)|^2 |v(y)|^2 \dix \diy \leq \norm{v}_{H^1}^4 o(1)
	\end{align}
	and
	\begin{align}\label{cv-ineq:hartree-nls-3body}
		0 \leq \norm{v}_{L^6}^6 - \iiint_{\R^6} W_{N^\beta}(x-y, x-z) |v(x)|^2 |v(y)|^2 |v(z)|^2 \dix \diy \diz \leq \norm{v}_{H^1}^6 o(1) \,,
	\end{align}
	with the $o(1)$’s independent of $v$.
	
	Assume in addition that $|x|U(x) \in L^1(\R^2)$. Then,
	\begin{equation}\label{cv-rate:hartree-nls-2body}
		0 \leq \norm{v}_{L^4}^4 - \iint_{\R^4} U_{N^\alpha}(x-y) |v(x)|^2 |v(y)|^2 \dix \diy \leq 2N^{-\alpha} \norm{ |\cdot| U(\cdot) }_{L^1} \norm{v}_{L^6}^3 \norm{ \nabla v }_{L^2} \,.
	\end{equation}
\end{lemma}

\begin{proof}
	The proof of~\eqref{cv-ineq:hartree-nls-2body} and~\eqref{cv-rate:hartree-nls-2body} can be found in~\cite{LewNamRou-17,LewNamRou-18-proc}. On the other hand,~\eqref{cv-ineq:hartree-nls-3body} is the 2D analogue of~\cite[Lemma 3.1]{NguRic-22-arxiv} and we omit its proof for brevity.
\end{proof}

With Lemma~\ref{lem:hartree-nls-2-3-body} in hand, we are now in the position to prove Theorem~\ref{thm:collapse-hartree}.

\begin{proof}[Proof of Theorem~\ref{thm:collapse-hartree}]
	The proof for approximate ground states being the same as the one for ground states, with very few changes, we only write the latter for brevity.
	
	First, we prove~\emph{\ref{thm:collapse-hartree_item1}}. That is, the convergence of Hartree energy in~\eqref{cv:hartree-nls-energy} as well as of its ground states in~\eqref{cv:hartree-nls-ground-state}. We see immediately from the variational principle, together with~\eqref{cv-ineq:hartree-nls-2body} and the nonnegativity in~\eqref{cv-ineq:hartree-nls-3body}, that
	\begin{equation}\label{upperbound:hartree-nls}
		\lim_{N\to+\infty}E_{a,b,N}^{\mathrm{H}} \leq E_{a,b}^{\NLS} \,.
	\end{equation}
	To prove the matching lower bound, we process as follows. Let $\{v_N\}_N$ be a sequence of ground states of~$E_{a,b,N}^{\mathrm{H}}$. We observe that if $\{\norm{\nabla v_N}_{L^2}\}_N$ is uniformly bounded, then it follows from~\eqref{upperbound:hartree-nls}, from the nonnegativity of the three\nobreakdash-body interaction, from the nonnegativity in~\eqref{cv-ineq:hartree-nls-2body}, and from~\eqref{ineq:GN} that $\{\norm{ |\cdot|^{s/2} v_N }_{L^2}\}_N$ is also uniformly bounded. For the uniform boundedness of $\{\norm{\nabla v_N}_{L^2}\}_N$, we treat the two cases $0< a < a_*$ and ($a \geq a_*$ and $\alpha < \beta$) for which we obtain the claim. First, we observe that if $0< a < a_*$, then the argument for $\{\norm{ |\cdot|^{s/2} v_N }_{L^2}\}_N$ just mentioned actually gives also the uniform boundedness of $\{\norm{\nabla v_N}_{L^2}\}_N$. Second, if $a \geq a_*$ with the additional assumption $\alpha < \beta$, then assume on the contrary that $\norm{\nabla v_N}_{L^2} \to+\infty$ as $N\to+\infty$. Define $\widetilde{v}_N = \varepsilon_N v_N(\varepsilon_N \cdot)$ with $\varepsilon_N := \norm{\nabla v_N}_{L^2}^{-1}$. The nonnegativity in~\eqref{cv-ineq:hartree-nls-2body} yields
	\begin{multline}\label{conv:hartree-nls-ground-state}
		E_{a,b,N}^{\mathrm{H}} = \cE_{a, b,N}^{\mathrm{H}}(v_N) \geq \varepsilon_N^{-4}\frac{b}{6}\iiint_{\R^6} W_{N^\beta\varepsilon_N}(x-y, x-z)|\widetilde{v}_N(x)|^2|\widetilde{v}_N(y)|^2|\widetilde{v}_N(z)|^2 \dix \diy \diz \\
		+ \varepsilon_N^{-2} \left( \norm{ \nabla \widetilde{v}_N }_{L^2}^2 - \frac{a}{2} \norm{ \widetilde{v}_N }_{L^4}^4 \right) + \varepsilon_N^s \int_{\R^2}|x|^{s}|\widetilde{v}_N(x)|^2 \dix \,.
	\end{multline}
	On the one hand, after multiplying~\eqref{conv:hartree-nls-ground-state} by $\varepsilon_N^2$ then using the nonnegativity of the three\nobreakdash-body interaction and of the external potential, equation~\eqref{upperbound:hartree-nls} allows to deduce that
	\begin{equation}\label{conv:hartree-nls-ground-state-0}
		\lim_{N\to+\infty}\norm{ \widetilde{v}_N }_{L^4}^4 \geq \frac{2}{a} > 0 \,,
	\end{equation}
	since $\norm{ \nabla \widetilde{v}_N }_{L^2} = 1$ by construction. Moreover, equation~\eqref{ineq:GN} and $\norm{ \widetilde{v}_N }_{L^2} = 1 = \norm{ \nabla \widetilde{v}_N }_{L^2}$ give
	\begin{equation}\label{conv:hartree-nls-ground-state-1}
		\norm{ \nabla \widetilde{v}_N }_{L^2}^2 -\frac{a}{2} \norm{ \widetilde{v}_N }_{L^4}^4 \geq 1 - \frac{a}{a_*} \,.
	\end{equation}
	On the other hand, multiplying~\eqref{conv:hartree-nls-ground-state} by $\varepsilon_N^4$ and using~\eqref{conv:hartree-nls-ground-state-1}, we obtain
	\[
		\lim_{N\to+\infty}\iiint_{\R^6}W_{N^\beta\varepsilon_N}(x-y, x-z)|\widetilde{v}_N(x)|^2|\widetilde{v}_N(y)|^2|\widetilde{v}_N(z)|^2 \dix \diy \diz = 0 \,,
	\]
	and claim that, due to the assumption~$\alpha<\beta$, it yields the strong convergence $\widetilde{v}_N \to 0$ in~$L^6(\R^2)$. Indeed, we first observe that $\varepsilon_N \geq CN^{-\alpha}$, which follows from~\eqref{upperbound:hartree-nls}, from the nonnegativity of the three\nobreakdash-body interaction and of the external potential, and from the fact that $\norm{ U_{N^\alpha} }_{L^\infty} = N^{2\alpha} \norm{U}_{L^\infty}$. Therefore, if~$\alpha<\beta$, then $N^{\beta}\varepsilon_N \geq CN^{\beta-\alpha} \to+\infty$ as $N\to+\infty$ and, consequently, the scaled three\nobreakdash-body interaction $W_{N^\beta\varepsilon_N}$ must converge to the delta interaction $\delta_{x=y=z}$. More precisely, replacing $W_{N^\beta}$ by $W_{N^\beta\varepsilon_N}$ in~\eqref{cv-ineq:hartree-nls-3body} ---the proof of which does not depend on the specific rate of divergence $N^\beta$, hence applies to~$N^\beta\varepsilon_N$--- and using crucially that the $o(1)$ in~\eqref{cv-ineq:hartree-nls-3body} does not depend on the $v\in H^1(\R^2)$, hence applies to~$\widetilde{v}_N$ for which $\norm{\widetilde{v}_N}_{H^1}^2=2$, we obtain
	\[
		\lim_{N\to+\infty} \norm{ \widetilde{v}_N }_{L^6}^6 = \lim_{N\to+\infty} \iiint_{\R^6}W_{N^\beta\varepsilon_N}(x-y, x-z)|\widetilde{v}_N(x)|^2|\widetilde{v}_N(y)|^2|\widetilde{v}_N(z)|^2 \dix \diy \diz = 0 \,.
	\]
	Consequently, $\widetilde{v}_N \to 0$ strongly in~$L^{p}(\R^2)$, for $2 < p \leq 6$, by interpolation, contradicting~\eqref{conv:hartree-nls-ground-state-0}.
	
	Therefore, we have proved in both our cases that $\{\norm{\nabla v_N}_{L^2}\}_N$ and $\{\norm{ |\cdot|^{s/2} v_N }_{L^2}\}_N$ are bounded uniformly. Then, there exists $v\in H^1(\R^2)$ such that, up to a subsequence, the convergence $v_N \to v$ holds weakly in~$H^1(\R^2)$, almost everywhere in~$\R^2$, and strongly in~$L^r(\R^2)$ for $2\leq r<+\infty$. In particular, $\norm{v}_{L^2}=1$ hence, by~\eqref{cv-ineq:hartree-nls-2body}--\eqref{cv-ineq:hartree-nls-3body} and the weak lower semicontinuity, we have
	\[
		\lim_{N\to+\infty}\cE_{a,b,N}^{\mathrm{H}}(v_N) \geq \cE_{a,b}^{\NLS}(v) \geq E_{a,b}^{\NLS} \,.
	\]
	Together with~\eqref{upperbound:hartree-nls}, this yields~\eqref{cv:hartree-nls-energy} as well as the $H^1(\R^2)$-convergence in~\eqref{cv:hartree-nls-ground-state}.
	
	\medskip
	
	Now we complete the proof of Theorem~\ref{thm:collapse-hartree} by proving~\emph{\ref{thm:collapse-hartree_item2}}. That is, the asymptotic behavior of Hartree energy and of its ground states in the collapse regime. By the variational principle, the second inequality in~\eqref{cv-rate:hartree-nls-2body} and the nonnegativity in~\eqref{cv-ineq:hartree-nls-3body}, we have
	\begin{align}\label{blow-up:hartree-upper-bound}
		E_{a_N, b_N, N}^{\mathrm{H}} &\leq \cE_{a_N, b_N, N}^{\mathrm{H}}(\ell_N^{-1}Q_0(\ell_N^{-1}\cdot)) \nonumber \\
		& \leq \cE_{a_N, b_N}^{\NLS}(\ell_N^{-1}Q_0(\ell_N^{-1}\cdot)) + CN^{-\alpha} \ell_N^{-3} = \left( \frac{1}{2} + \frac{1}{s} - \frac{\zeta}{4} + CN^{-\alpha} \ell_N^{-s-3} + o(1)\right) \cQ_s \ell_N^s \,.
	\end{align}
	We recall that we assume~$\ell_N = N^{-\eta}$ with $\eta > 0$. Thus, the error term~$N^{-\alpha} \ell_N^{-s-3}$ is negligible when~$\eta<\alpha/(s+3)$ and we obtained the upper bound in~\eqref{blowup:hartree-inhomogeneous-energy}. The matching lower bound in~\eqref{blowup:hartree-inhomogeneous-energy} is a consequence of the claimed convergence~\eqref{blowup:hartree-inhomogeneous-ground-state}, which is obtained as follows. Let $v_N$ be a ground state of~$E_{a_N, b_N, N}^{\mathrm{H}}$ and $w_N := \ell_N v_N(\ell_N \cdot)$. Then, $\norm{w_N}_{L^2} = \norm{v_N}_{L^2}=1$ and, by the nonnegativity in~\eqref{cv-ineq:hartree-nls-2body},
	\begin{multline}\label{blow-up:hartree-lower-bound}
		E_{a_N, b_N, N}^{\mathrm{H}} \geq \ell_N^{-4}\frac{b_N}{6}\iiint_{\R^6} W_{N^\beta\ell_N}(x-y, x-z)|w_N(x)|^2|w_N(y)|^2|w_N(z)|^2 \dix \diy \diz \\
		+ \ell_N^{-2}\left(\norm{\nabla w_N}_{L^2}^2 - \frac{a_N}{2}\norm{w_N}_{L^4}^4 \right) + \ell_N^s \int_{\R^2}|x|^{s}|w_N(x)|^2 \dix \,.
	\end{multline}
	We distinguish two (overlapping) cases for~$\zeta \geq 0$: $\zeta\neq0$ and $\zeta\neq1$.
	
	\medskip
	\textbf{Case $\zeta\neq0$.}
	Combining~\eqref{blow-up:hartree-upper-bound}, under the assumption~$\eta<\alpha/(s+3)$, with~\eqref{blow-up:hartree-lower-bound} and using the definition of~$\ell_N$ depending on $b_N$ in~\eqref{thm:collapse-nls-blowup-length}, we obtain
	\begin{multline}\label{blow-up:hartree-1}
		\frac{1}{2} + \frac{1}{s} - \frac{\zeta}{4} + o(1) \geq \frac{\zeta}{24 \cQ_6} \iiint_{\R^6} W_{N^\beta\ell_N}(x-y, x-z)|w_N(x)|^2|w_N(y)|^2|w_N(z)|^2 \dix \diy \diz \\
		+ \ell_N^{-s-2} \cQ_s^{-1} \left(\norm{\nabla w_N}_{L^2}^2 - \frac{a_N}{2}\norm{w_N}_{L^4}^4 \right) + \cQ_s^{-1} \int_{\R^2}|x|^{s}|w_N(x)|^2 \dix \,.
	\end{multline}
	Moreover, by~\eqref{ineq:GN},~\eqref{thm-collapse-nls-condition-ratio-sequences}, and the definition of~$\ell_N$ depending on $b_N$ in~\eqref{thm:collapse-nls-blowup-length}, we have
	\begin{equation}\label{blow-up:hartree-2}
		\ell_N^{-s-2} \cQ_s^{-1} \left(\norm{\nabla w_N}_{L^2}^2 - \frac{a_N}{2}\norm{w_N}_{L^4}^4 \right) \geq \left(\frac{1-\zeta}{2} + o(1)\right)\norm{\nabla w_N}_{L^2}^2 \,.
	\end{equation}
	Again, we observe that if $\{\norm{\nabla w_N}_{L^2}\}_N$ is uniformly bounded, then it follows from~\eqref{blow-up:hartree-1} and~\eqref{blow-up:hartree-2} that $\{\norm{ |\cdot|^{s/2} w_N }_{L^2}\}_N$ is also uniformly bounded. For the uniform boundedness of $\{\norm{\nabla w_N}_{L^2}\}_N$, we treat the two cases $0< \zeta < 1$ and ($\zeta \geq 1$ and $\alpha < \beta$) for which we obtain the claim. First, we observe that if $0< \zeta < 1$, then the argument for $\{\norm{ |\cdot|^{s/2} w_N }_{L^2}\}_N$ just mentioned actually gives also the uniform boundedness of $\{\norm{\nabla w_N}_{L^2}\}_N$. Second, if $\zeta \geq 1$ with the additional assumption $\alpha < \beta$, then assume on the contrary that $\norm{\nabla w_N}_{L^2} \to+\infty$ as $N\to+\infty$. Define $\widetilde{w}_N = \varepsilon_N w_N(\varepsilon_N \cdot)$ with $\varepsilon_N := \norm{\nabla w_N}_{L^2}^{-1} \to 0$. Dropping the nonnegative external term in~\eqref{blow-up:hartree-1}, we have
	\begin{multline}\label{blow-up:hartree-3}
		\frac{1}{2} + \frac{1}{s} - \frac{\zeta}{4} + o(1) \geq \varepsilon_N^{-4}\frac{\zeta}{24 \cQ_6} \iiint_{\R^6}W_{N^{\beta}\ell_N \varepsilon_N}(x-y, x-z)|\widetilde{w}_N(x)|^2|\widetilde{w}_N(y)|^2|\widetilde{w}_N(z)|^2 \dix \diy \diz \\
		+ \varepsilon_N^{-2}\ell_N^{-s-2} \cQ_s^{-1} \left( \norm{ \nabla \widetilde{w}_N }_{L^2}^2 - \frac{a_N}{2} \norm{ \widetilde{w}_N }_{L^4}^4 \right).
	\end{multline}
	On the one hand, multiplying~\eqref{blow-up:hartree-3} by $\varepsilon_N^2 \ell_N^{s+2}$, dropping the nonnegative three\nobreakdash-body term, and using that $\norm{ \widetilde{w}_N }_{L^2}=1=\norm{ \nabla \widetilde{w}_N }_{L^2}$ and that $a_N \to a_*$ as $N\to+\infty$, we deduce that
	\begin{equation}\label{blow-up:hartree-4}
		\lim_{N\to+\infty}\norm{ \widetilde{w}_N }_{L^4}^4 \geq \frac{2}{a_*} > 0 \,.
	\end{equation}
	On the other hand, multiplying~\eqref{blow-up:hartree-3} by $\varepsilon_N^4$ and using~\eqref{blow-up:hartree-2} on $\widetilde{w}_N$, we deduce that
	\begin{equation}\label{blow-up:hartree-5}
		\lim_{N\to+\infty}\iiint_{\R^6}W_{N^{\beta}\ell_N \varepsilon_N}(x-y, x-z)|\widetilde{w}_N(x)|^2|\widetilde{w}_N(y)|^2|\widetilde{w}_N(z)|^2 \dix \diy \diz = 0 \,.
	\end{equation}
	Since $v_N$ is a ground state of~$E_{a_N, b_N, N}^{\mathrm{H}}$, similarly to earlier we can drop the nonnegativity three\nobreakdash-body and external terms and use $\norm{ U_{N^\alpha} }_{L^\infty} = N^{2\alpha} \norm{U}_{L^\infty}$ for a lower bound on~$E_{a_N, b_N, N}^{\mathrm{H}}$ in~\eqref{blow-up:hartree-upper-bound}. Passing then to the limit yields $\norm{\nabla v_N}_{L^2} \lesssim N^\alpha$. Thus, $\varepsilon_N = \norm{\nabla w_N}_{L^2}^{-1} = \ell_N^{-1}\norm{\nabla v_N}_{L^2}^{-1} \geq C \ell_N^{-1}N^{-\alpha}$ and, when $\alpha<\beta$, $N^{\beta}\ell_N \varepsilon_N \geq C N^{\beta-\alpha} \to+\infty$ as $N\to+\infty$.  Therefore, with the same arguments as before, replacing $W_{N^\beta}$ by $W_{N^{\beta}\ell_N \varepsilon_N}$ in~\eqref{cv-ineq:hartree-nls-3body} and using~\eqref{blow-up:hartree-5} gives that $\widetilde{w}_N \to 0$ strongly in~$L^6(\R^2)$. Consequently, $\widetilde{w}_N \to 0$ strongly in~$L^{p}(\R^2)$ for $2<p\leq 6$, by interpolation, contradicting~\eqref{blow-up:hartree-4}.
	
	Therefore, we have proved in both our cases that $\{\norm{\nabla w_N}_{L^2}\}_N$ and $\{\norm{ |\cdot|^{s/2} w_N }_{L^2}\}_N$ are bounded uniformly. Thus, there exists $w\in H^1(\R^2)$ such that, up to a subsequence, the convergence $w_N \to w$ holds weakly in~$H^1(\R^2)$, almost everywhere in~$\R^2$, and strongly in~$L^r(\R^2)$ for $2\leq r < +\infty$. At this stage, we note that $N^\beta\ell_N \to+\infty$ as $N\to+\infty$ , due to our assumptions $\ell_N \sim N^{-\eta}$ and $\eta < \beta$. We then use~\eqref{cv-ineq:hartree-nls-3body} on $w_N$, with the uniform boundedness of $\{\norm{\nabla w_N}_{L^2}\}_N$ and the strong convergence $w_N \to w$ in~$L^6(\R^2)$, and obtain
	\[
		\lim_{N\to+\infty}\iiint_{\R^6} W_{N^\beta\ell_N}(x-y, x-z)|w_N(x)|^2|w_N(y)|^2|w_N(z)|^2 \dix \diy \diz = \norm{w}_{L^6}^6 \,.
	\]
	Inserting the above into~\eqref{blow-up:hartree-1}, using~\eqref{blow-up:hartree-2}, and adapting arguments in the proof of Theorem~\ref{thm:collapse-nls}, we conclude the proof of the convergence of Hartree ground states in~\eqref{blowup:hartree-inhomogeneous-ground-state} and of the energy lower bound in~\eqref{blowup:hartree-inhomogeneous-energy}.
	
	\medskip
	
	\textbf{Case $\zeta\neq1$.} In this case, $\frac{A_N}{1-\zeta}>0$ (recall that $A_n := 1 - a_n/a_*$) for $N$ large enough, by~\eqref{thm-collapse-nls-condition-ratio-sequences}. Combining~\eqref{blow-up:hartree-upper-bound} with~\eqref{blow-up:hartree-lower-bound} and using~\eqref{ineq:GN} and the definition of~$\ell_N$ depending on $A_N$ in~\eqref{thm:collapse-nls-blowup-length} ---it is well defined at least for $N$ large enough since, by~\eqref{thm-collapse-nls-condition-ratio-sequences}, $A_N/(1-\zeta) > 0$ for $N$ large enough---, we obtain
	
	Applying~\eqref{ineq:GN} to the $L^4$-norm term in~\eqref{blow-up:hartree-lower-bound} and rewriting the lower bound obtained that way, using~\eqref{thm-collapse-nls-condition-ratio-sequences} and the definition of~$\ell_N$ depending on $A_N= 1 - a_N/a_*$ in~\eqref{thm:collapse-nls-blowup-length} ---it is well defined at least for $N$ large enough since, by~\eqref{thm-collapse-nls-condition-ratio-sequences}, $A_N/(1-\zeta) > 0$ for $N$ large enough---, we obtain by~\eqref{blow-up:hartree-upper-bound}, where we assume $\eta<\alpha/(s+3)$, that
	
	\begin{multline*}
		\frac{1}{2} + \frac{1}{s} - \frac{\zeta}{4} + o(1) \geq \left(\frac{\zeta}{24 \cQ_6}+o(1)\right) \iiint_{\R^6} W_{N^\beta\ell_N}(x-y, x-z) |w_N(x)|^2 |w_N(y)|^2 |w_N(z)|^2 \dix \diy \diz\\
		+ \frac{1-\zeta}{2} \norm{\nabla w_N}_{L^2}^2 + \cQ_s^{-1}\int_{\R^2}|x|^{s} |w_N(x)|^2 \dix \,.
	\end{multline*}
	For $0\leq \zeta<1$ the above yields the uniformly boundedness of~$\{\norm{\nabla w_N}_{L^2}\}_N$ and $\{\norm{ |\cdot|^{s/2} w_N }_{L^2}\}_N$. This actually holds true when $\zeta \geq 1$ too, under the additional assumption $\alpha<\beta$, by the same arguments as in the case $\zeta \neq 0$. For the rest of the proof, we omit the details since it strictly follows the end of the proof in the case $\zeta \neq 0$.
\end{proof}

\subsection{Condensation and collapse in the many-body theory}

The goal of this subsection is to prove Theorem~\ref{thm:many-body}. The strategy and main difficulty is to compare the quantum energy $E_{a,b,N}^{\mathrm{QM}}$ in~\eqref{energy:quantum} and the Hartree energy $E_{a,b,N}^{\mathrm{H}}$ in~\eqref{energy:hartree-inhomogeneous}. Comparing the quantum energy and the Hartree energy is a common strategy in many\nobreakdash-body quantum systems to understand the role of interactions and the validity of the mean-field approximation. It is worth noting that the main difficulty lies in the energy lower bound since the upper bound can be obtained by using a factorized state. Once those bounds are obtained, the result then comes from Theorem~\ref{thm:collapse-hartree} where we established the condensation of the Hartree ground states in the limit $N\to+\infty$ and investigated its blowup profile.

\subsubsection{Convergence and collapse of the quantum energy}

A major ingredient in our proof is the information-theoretic quantum de~Finetti theorem from~\cite{LiSmi-15,BraHar-17}. The following formulation is taken for its first part from~\cite[Theorem 3.7]{Rougerie-20b} (where a general discussion and more references can be found) and from~\cite[Proof of Lemma 19]{NamRicTri-23} for its second part.
\begin{theorem}[Information-theoretic quantum de~Finetti]\label{thm:information_deF}
	Let $\gH$ be a complex separable Hilbert space and $\gH_N = \gH^{\otimes_{\mathrm{sym}} N}$ the corresponding bosonic space. Let $\gamma_{\Psi_N}^{(3)}$ be the $3$-body reduced density matrix of a $N$-body state vector $\Psi_N \in \gH_N$ and $P$ be a finite dimensional orthogonal projector. Then, there exists a Borel measure~$\mu_{\Psi_N}$ with total mass $\leq 1$ on the set of one-body mixed states
	\[
		\cS_{P} := \left\{ \gamma \mbox{ positive, trace-class, self-adjoint operator on } P\gH, \, \tr \gamma = 1 \right\}
	\]
	such that
	\begin{equation}\label{ineq:information-deF}
		\sup_{0\leq A,B,C\leq 1}\tr\left| A \otimes B \otimes C \left( P^{\otimes 3} \gamma_{\Psi_N}^{(3)} P^{\otimes 3} - \int_{\cS_{P}} \gamma^{\otimes 3} \di\mu_{\Psi_N}(\gamma) \right)\right| \leq C \sqrt{\frac{\log (\dim P)}{N}}
	\end{equation}
	where the supremum is taken over bounded operators on $P\gH$. Furthermore, with $P_{\perp}=\1-P$,
	\begin{equation}\label{ineq:measure-deF}
		1\geq \int_{\cS_{P}} \di\mu_{\Psi_N}(\gamma) = \tr\left[P^{\otimes 3}\gamma_{\Psi_N}^{(3)} P^{\otimes 3}\right] \geq 1- 3\tr\left[ P_{\perp} \gamma_{\Psi_N}^{(1)} \right].
	\end{equation}
\end{theorem}

We will apply Theorem~\ref{thm:information_deF} with $P$ a spectral projector below an energy cut-off $L > 0$ for the one-body operator $h_x = -\Delta_x + |x|^s$, i.e.,
\[
	P := \1(h\leq L) \,.
\]
Note that the projected Hilbert space $P\gH$ is of finite dimension, which follows from the celebrated semi-classical inequality ``{\`a} la Cwikel--Lieb--Rozenblum'' (see, e.g.,~\cite[Lemma 3.3]{LewNamRou-16}):
\begin{equation}\label{CLR}
	\dim(P\gH)\lesssim L^{1+\frac{2}{s}} \,.
\end{equation}

We project through $P$ the many\nobreakdash-body ground states $\Psi_N$ onto the finite dimensional space and bound the full energy from below in terms of projected state of the Hamiltonian~$H_{a,b,N}$. We write
\begin{align}\label{lower-bound:many-body}
	\frac{\langle \Psi_N, H_{a,b,N} \Psi_N \rangle}{N} &= \tr\left[ h \gamma_{\Psi_N}^{(1)} \right] - \frac{a}{2}\tr\left[U_{N^\alpha}\gamma_{\Psi_N}^{(2)} \right] + \frac{b}{6}\tr\left[W_{N^\beta}\gamma_{\Psi_N}^{(3)} \right] \nonumber \\
	&= \tr\left[\left(\frac{h_1+h_2+h_3}{3} - \frac{a}{4}\cU_{N^\alpha} + \frac{b}{6}W_{N^\beta}\right)\gamma_{\Psi_N}^{(3)} \right] =: \frac{1}{3}\tr\left[H_{a,b,N}^{(3)}\gamma_{\Psi_N}^{(3)}\right],
\end{align}
where $h_i$ acts on the $i$-th variable and where we defined
\[
	\cU_{N^\alpha} := U_{N^\alpha}\otimes \1 + \1 \otimes U_{N^\alpha} \,.
\]
For the one-body noninteracting term, we have
\begin{align}
	\tr\left[ h \gamma_{\Psi_N}^{(1)} \right] ={} &
	\tr\left[PhP\gamma_{\Psi_N}^{(1)} \right] + \tr\left[P_{\perp}hP_{\perp}\gamma_{\Psi_N}^{(1)} \right] \nonumber\\
	\geq{} &
	\frac{1}{3}\tr\left[P^{\otimes 3}(h_1+h_2+h_3)P^{\otimes 3}\gamma_{\Psi_N}^{(3)} \right] + L\tr\left[ P_{\perp} \gamma_{\Psi_N}^{(1)} \right].\label{lower-bound:one-body}
\end{align}
For the two\nobreakdash-body interaction term, we write
\[
	\cU_{N^\alpha} = P^{\otimes 3}\cU_{N^\alpha}P^{\otimes 3} + \Pi\cU_{N^\alpha}\Pi + P^{\otimes 3}\cU_{N^\alpha}\Pi + \Pi\cU_{N^\alpha}P^{\otimes 3}
\]
with the orthogonal projection
\begin{align}
	\Pi := \1^{\otimes 3} - P^{\otimes 3} ={} & P_{\perp} \otimes \1 \otimes \1 + P \otimes P_{\perp} \otimes 1 + P \otimes P \otimes P_{\perp} \nonumber \\
	\leq{} & P_{\perp} \otimes \1 \otimes \1 + \1 \otimes P_{\perp} \otimes 1 + \1 \otimes \1 \otimes P_{\perp}.\label{lower-bound:error-body}
\end{align}
To bound the error terms, we use the inequality
\begin{equation}\label{lower-bound:operator}
	XAY + YAX \geq -\varepsilon X|A|X - \varepsilon^{-1} Y|A|Y
\end{equation}
valid for any $\varepsilon > 0$, any self-adjoint operator $A$, and any orthogonal projectors $X$, $Y$. Applying~\eqref{lower-bound:operator} to $X=P^{\otimes 3}$, $Y=\Pi$ and $A=-\cU_{N^\alpha}\leq0$, we deduce that
\begin{align}\label{lower-bound:two-body-operator}
	\cU_{N^\alpha} \leq{} & (1+\varepsilon)P^{\otimes 3}\cU_{N^\alpha}P^{\otimes 3} + (1+\varepsilon^{-1})\Pi\cU_{N^\alpha}\Pi \\
	\leq{} & P^{\otimes 3}\cU_{N^\alpha}P^{\otimes 3} + CN^{2\alpha} \left(\varepsilon P^{\otimes 3} + (1+\varepsilon^{-1})\Pi\right),\nonumber
\end{align}
where we have used that $U_{N^\alpha} \leq N^{2\alpha} \norm{U}_{L^\infty}$. Using~\eqref{lower-bound:error-body}, then taking the trace of~\eqref{lower-bound:two-body-operator} against $\gamma_{\Psi_N}^{(3)}$, then using $\tr\left[ P^{\otimes 3}\gamma_{\Psi_N}^{(3)} \right] \leq 1$, then optimizing over $\varepsilon$, and finally using $\tr\left[ P_{\perp} \gamma_{\Psi_N}^{(1)} \right] \leq 1$ gives
\begin{equation}\label{lower-bound:two-body}
	\tr\left[\cU_{N^\alpha}\gamma_{\Psi_N}^{(3)}\right] \leq\tr\left[\cU_{N^\alpha}P^{\otimes 3}\gamma_{\Psi_N}^{(3)}P^{\otimes 3}\right] + CN^{2\alpha}\sqrt{\tr\left[ P_{\perp} \gamma_{\Psi_N}^{(1)} \right]} \,.
\end{equation}
Similarly, applying~\eqref{lower-bound:operator} to $X=P^{\otimes 3}$, $Y=\Pi$, $A=W_{N^\beta}\geq0$ and using $W_{N^\beta} \leq N^{4\beta} \norm{W}_{L^\infty}$, we obtain
\begin{equation}\label{lower-bound:three-body}
	\tr\left[W_{N^\beta}\gamma_{\Psi_N}^{(3)} \right] \geq \tr\left[W_{N^\beta}P^{\otimes 3}\gamma_{\Psi_N}^{(3)}P^{\otimes 3} \right] - CN^{4\beta}\sqrt{\tr\left[ P_{\perp} \gamma_{\Psi_N}^{(1)} \right]} \,.
\end{equation}
Putting together the estimates~\eqref{lower-bound:one-body},~\eqref{lower-bound:two-body}, and~\eqref{lower-bound:three-body} into~\eqref{lower-bound:many-body}, we arrive at
\begin{align}\label{lower-bound:truncated}
	\frac{ \langle \Psi_N, H_{a,b,N} \Psi_N \rangle }{N} \geq{} & \frac{1}{3} \tr\left[ H_{a,b,N}^{(3)}P^{\otimes3}\gamma_{\Psi_N}^{(3)}P^{\otimes3} \right] + L\tr\left[ P_{\perp} \gamma_{\Psi_N}^{(1)} \right] - C\left(N^{2\alpha}+N^{4\beta}\right) \sqrt{\tr\left[ P_{\perp} \gamma_{\Psi_N}^{(1)} \right]} \nonumber \\
	\geq{} & \frac{1}{3} \tr\left[ H_{a,b,N}^{(3)}P^{\otimes3}\gamma_{\Psi_N}^{(3)}P^{\otimes3} \right] + \frac{L}{2}\tr\left[ P_{\perp} \gamma_{\Psi_N}^{(1)} \right] - C\frac{N^{4\alpha}+N^{8\beta}}{L} \,,
\end{align}
where the last inequality is obtained by the Cauchy--Schwarz inequality.

Next, by Theorem~\ref{thm:information_deF} we bound the truncated energy from below in terms of the de~Finetti measure. Using~\eqref{ineq:information-deF} and the triangle inequality on each of the three terms of~$H_{a,b,N}^{(3)}$ separately, together with $Ph \leq LP$, $\norm{ U_{N^\alpha} }_{L^\infty} = N^{2\alpha} \norm{U}_{L^\infty}$, and $\norm{ W_{N^\beta} }_{L^\infty} = N^{4\beta} \norm{W}_{L^\infty}$, respectively, then patching things together and recalling~\eqref{CLR}, we obtain
\begin{equation}\label{lower-bound:deF}
	\frac{1}{3} \tr\left[H_{a,b,N}^{(3)}P^{\otimes3}\gamma_{\Psi_N}^{(3)}P^{\otimes3} \right] \geq \frac{1}{3} \tr\left[H_{a,b,N}^{(3)}\int_{\cS_{P}} \gamma^{\otimes 3} \di\mu_{\Psi_N}(\gamma)\right] - C \sqrt{\frac{\log L}{N}} \left( L+N^{2\alpha}+N^{4\beta}\right).
\end{equation}
Moreover, the Hoffmann--Ostenhof inequality~\cite{HofHof-77}, $\tr\left[ h\gamma \right] \geq \int |\nabla \sqrt{\rho_{\gamma}}|^2$ where $\rho_{\gamma}(x)=\gamma(x,x)$, gives
\begin{align}\label{lower-bound:main}
	\frac{1}{3} \tr\left[H_{a,b,N}^{(3)}\gamma^{\otimes 3}\right] &=
	\begin{multlined}[t][0.78\textwidth]
		\tr\left[ h\gamma \right] - \frac{a}{2} \iint_{\R^4} N^{2\alpha} U(N^{\alpha}(x-y)) \rho_{\gamma}(x)\rho_{\gamma}(y) \dix \diy \\
		+ \frac{b}{6} \iiint_{\R^6} N^{4\beta} W(N^{\beta}(x-y),N^{\beta}(x-z)) \rho_{\gamma}(x)\rho_{\gamma}(y)\rho_{\gamma}(z) \dix \diy \diz
	\end{multlined} \nonumber\\
	&=: \cE_{a,b,N}^{\mathrm{mH}}(\gamma) \geq \cE_{a,b,N}^{\mathrm{H}}(\sqrt{\rho_{\gamma}}) \geq E_{a,b,N}^{\mathrm{H}} \,.
\end{align}
Combining~\eqref{lower-bound:truncated}--\eqref{lower-bound:main}, and since the $\Psi_N$'s are ground states of $E_{a,b,N}^{\mathrm{QM}}$, we obtain
\begin{multline}\label{cv:energy-qm-hartree-preparative}
	E_{a,b,N}^{\mathrm{QM}} = \frac{ \langle \Psi_N, H_{a,b,N} \Psi_N \rangle }{N} \geq E_{a,b,N}^{\mathrm{H}}\int_{\cS_{P}}\di\mu_{\Psi_N}(\gamma) + \frac{L}{2}\tr\left[ P_{\perp} \gamma_{\Psi_N}^{(1)} \right] \\
	- C\frac{N^{4\alpha}+N^{8\beta}}{L} - C \sqrt{\frac{\log L}{N}} \left( L+N^{2\alpha}+N^{4\beta}\right).
\end{multline}
We now choose $L$ as a power of $N$ so that the last two terms vanish as $N\to+\infty$: we need $N^{\frac{1}{2}} \gg L \gg N^{4\alpha}+N^{8\beta}$. Choosing optimally
\begin{equation}\label{eq:L}
	L \sim N^{2\alpha+\frac{1}{4}} + N^{4\beta+\frac{1}{4}} \quad \text{with} \quad \alpha<\frac{1}{8} \quad \text{and} \quad \beta<\frac{1}{16} \,,
\end{equation}
using either the lower bound or the upper bound in~\eqref{ineq:measure-deF} depending on $E_{a,b,N}^{\mathrm{H}}$ being nonnegative or positive, respectively, and using in the latter case that $E_{a,b,N}^{\mathrm{H}}$ is bounded (in order to bound from below $L/2 - 3 E_{a,b,N}^{\mathrm{H}}$ by $CL$), we obtain
\begin{equation}\label{cv:energy-qm-hartree-preparative2}
	E_{a,b,N}^{\mathrm{H}} \geq E_{a,b,N}^{\mathrm{QM}} \geq E_{a,b,N}^{\mathrm{H}} + C L \tr\left[ P_{\perp} \gamma_{\Psi_N}^{(1)} \right] - C\left(N^{2\alpha-\frac{1}{4}}+N^{4\beta-\frac{1}{4}}\right).
\end{equation}
The nonnegativity of the trace term gives the final estimate
\begin{equation}\label{cv:energy-qm-hartree}
	E_{a,b,N}^{\mathrm{H}} \geq E_{a,b,N}^{\mathrm{QM}} \geq E_{a,b,N}^{\mathrm{H}} - C\left(N^{2\alpha-\frac{1}{4}}+N^{4\beta-\frac{1}{4}}\right),
\end{equation}
which, with~\eqref{cv:hartree-nls-energy}, yields the claimed convergence~\eqref{cv:qm-nls-energy} of the quantum energy to the \NLS energy.

Finally, and because it will be useful later, note that the uniform boundedness of $E_{a,b,N}^{\mathrm{H}}$, the choice~\eqref{eq:L}, the nonnegativity of $C L \tr\left[ P_{\perp} \gamma_{\Psi_N}^{(1)} \right]$, and~\eqref{cv:energy-qm-hartree-preparative2} actually give
\begin{equation}\label{cv:trace_Pperp_one-body}
	\tr\left[P_{\perp}\gamma_{\Psi_N}^{(1)}\right] = o\!\left(L^{-1}\right) \xrightarrow[N\to+\infty]{} 0 \,,
\end{equation}
hence,~\eqref{ineq:measure-deF} implies that
\begin{equation}\label{cv:energy_measure}
	\lim_{N\to+\infty}\int_{\cS_{P}} \di\mu_{\Psi_N}(\gamma) = 1
\end{equation}
and the sequence~$\{\mu_{\Psi_N}\}_N$ given by Theorem~\ref{thm:information_deF} is tight on the set of one-body mixed states
\[
	\cS := \left\{ \gamma \mbox{ nonnegative, trace-class, self-adjoint operator on } L^2 (\R), \, \tr \gamma = 1 \right\}.
\]
Therefore, modulo a subsequence, the sequence of measures~$\{\mu_{\Psi_N}\}_N$ converges to a measure~$\mu$.

\medskip

In the collapse regime where $a_N \to a_*$ and $b_N \searrow 0$ at the same time, when $N\to+\infty$, the asymptotic formula of the quantum energy~\eqref{blowup:qm-energy} follows from~\eqref{cv:energy-qm-hartree} and~\eqref{blowup:hartree-inhomogeneous-energy}, with the notable fact that, in addition to the condition on the speed of collapse $\eta < \min\left\{ \alpha/(s+3), \beta \right\}$ in~\eqref{collapse:speed-hartree}, we furthermore need that $\eta < \min\left\{ 1-8\alpha, 1-16\beta \right\} / (4s)$ ---yielding therefore the condition~\eqref{collapse:speed-many-body}--- in order for the error terms in~\eqref{cv:energy-qm-hartree} to be of small order compared to the \NLS energy~\eqref{blowup:nls-energy}.

\subsubsection{Convergence and collapse of many\nobreakdash-body ground states}
We prove now the condensation and collapse of many\nobreakdash-body ground states, given in~\eqref{cv:qm-nls-ground-state} and~\eqref{blowup:qm-ground-state}. Note that, when the limiting profile is unique, this can be obtained by the Feynman--Hellmann principle (see, e.g.,~\cite{LewNamRou-18-proc}). This is the case in the collapse regime since the critical cubic \NLS equation~\eqref{eq:NLS} admits a unique solution, which is the limiting profile. In the mean-field regime, however, we need to employ the quantum de~Finetti theory and an abstract argument of convex analysis to deal with the nonuniqueness of the cubic--quintic \NLS ground states of~\eqref{energy:nls-inhomogeneous}.

In the mean-field regime, we recall that $a,b>0$ are fixed, that $0 < \alpha < 1/8$, and that $0 < \beta < 1/16$, with either $a < a_*$ or ($a\geq a_*$ and $\alpha<\beta$). Let $\{\Psi_N\}_N$ be a sequence of ground states of~$H_{a,b,N}$.

We first prove that $\tr\left[ h \gamma_{\Psi_N}^{(1)} \right]$ is bounded uniformly, where we recall that $h_x=-\Delta_x + |x|^s$ is the one-body noninteracting operator. To this purpose, we introduce the perturbed Hamiltonian
\[
	H_{a,b,N,\lambda} = H_{a,b,N} - \lambda\sum_{i=1}^N h_{x_i}
\]
with the corresponding quantum energy $E_{a,b,N,\lambda}^{\mathrm{QM}}$. Here $\lambda$ is fixed with $0<\lambda < 1 - a/a_*$ if $a < a_*$, and with $0<\lambda<1$ if $a\geq a_*$ and $\alpha<\beta$, in order to ensure that $a/(1-\lambda)$ compares the same way to $a_*$ as $a$ does.
We return to~\eqref{lower-bound:truncated} and~\eqref{lower-bound:deF} to derive an energy lower bound similar to~\eqref{cv:energy-qm-hartree} with $E_{a,b,N}^{\mathrm{QM}}$ replaced by $E_{a,b,N,\lambda}^{\mathrm{QM}}$. The derivation is the same, except for~\eqref{lower-bound:main} where the lower bound is now $(1-\lambda)E_{a/(1-\lambda), b/(1-\lambda), N}^{\mathrm{H}}$ instead of $E_{a,b,N}^{\mathrm{H}}$, and we find in particular that $H_{a,b,N,\lambda} \geq -C_{\lambda}N$. Thus, the desired uniform boundedness of~$\tr\left[ h \gamma_{\Psi_N}^{(1)} \right]$ since
\[
	E_{a,b}^{\NLS} + o(1) \geq \frac{ \langle \Psi_N, H_{a,b,N} \Psi_N \rangle }{N} \geq -C_{\lambda} + \lambda\tr\left[ h \gamma_{\Psi_N}^{(1)} \right].
\]
Therefore, and since $h$ has compact resolvent we deduce that, modulo a subsequence, $\gamma_{\Psi_N}^{(1)}$ converges strongly in the trace-class.

By the de~Finetti theorem~\cite[Theorem~2.1]{LewNamRou-14}, there exists a unique Borel probability measure~$\nu$ on the unit sphere $S\gH$ such that
\begin{equation}\label{eq:cv-quantum-de-finetti}
	\lim_{N\to+\infty}\tr\left| \gamma_{\Psi_N}^{(k)} - \int_{S\gH}|v^{\otimes k}\rangle\langle v^{\otimes k}|\di\nu(v)\right| = 0 \,,\quad \forall k\in\N \,.
\end{equation}
Therefore, to complete the proof of~\eqref{cv:qm-nls-ground-state}, it remains to prove that $\nu$ in~\eqref{eq:cv-quantum-de-finetti} is supported on the set of cubic--quintic \NLS ground states $\cM^{\NLS}$.
To this end, we derive an alternative lower bound to that in~\eqref{cv:energy-qm-hartree} where we replace $E_{a,b,N}^{\mathrm{H}}$ by the modified Hartree functional defined on mixed states $\cE_{a,b,N}^{\mathrm{mH}}$. This is simply done by omitting the inequalities in~\eqref{lower-bound:main} and instead keeping $\cE_{a,b,N}^{\mathrm{mH}}(\gamma)$ for the rest of the derivation. We obtain
\begin{equation}\label{cv:ground-states-deF}
	E_{a,b,N}^{\mathrm{H}} \geq E_{a,b,N}^{\mathrm{QM}} \geq \int_{\cS_{P}} \cE_{a,b,N}^{\mathrm{mH}}(\gamma)\di\mu_{\Psi_N} (\gamma) - o(1) \geq \int_{\cS_{P}}\cE_{a,b,N}^{\mathrm{H}}(\sqrt{\rho_\gamma})\di\mu_{\Psi_N} (\gamma) - o(1)
\end{equation}
and we split the integral into two parts: low and high kinetic energy
\[
	\mathrm{K}_- = \{\gamma\in \cS: \langle \sqrt{\rho_{\gamma}},h\sqrt{\rho_{\gamma}} \rangle \leq C_{\mathrm{kin}}\} \quad \text{and} \quad \mathrm{K}_+ = \cS \setminus \mathrm{K}_- \,,
\]
for $C_{\mathrm{kin}} > 0$ a large constant independent of~$N$. On one hand, for the high kinetic energy part, we consider the modified Hartree functional
\[
	\cE_{a,b,N,\lambda}^{\mathrm{H}}(u) = \cE_{a,b,N}^{\mathrm{H}}(u) - \lambda \langle u,hu \rangle\,,
\]
with the corresponding modified Hartree energy $E_{a,b,N,\lambda}^{\mathrm{H}}$. Recall that $ \cE_{a,b,N}^{\mathrm{H}}$ was defined in~\eqref{functional:hartree-inhomogeneous}. Again, and for the same reason, $\lambda$ is fixed with $0<\lambda < 1 - a/a_*$ if $a < a_*$, and with $0<\lambda<1$ if $a\geq a_*$ and $\alpha<\beta$. We return to Theorem~\ref{thm:collapse-hartree} and derive an energy lower bound similar to~\eqref{cv:hartree-nls-energy} with $E_{a,b,N}^{\mathrm{H}}$ replaced by $E_{a,b,N,\lambda}^{\mathrm{H}}$. Again, we find in particular that $E_{a,b,N,\lambda}^{\mathrm{H}} \geq -C_\lambda$ and deduce that
\begin{equation}\label{cv:ground-states-split-1}
	\cE_{a,b,N}^{\mathrm{mH}}(\gamma) \geq \cE_{a,b,N}^{\mathrm{H}}(\sqrt{\rho_{\gamma}}) \geq \lambda \langle \sqrt{\rho_{\gamma}},h\sqrt{\rho_{\gamma}} \rangle - C_{\lambda} \geq \frac{\lambda C_{\mathrm{kin}}}{2},\quad \forall \gamma \in \mathrm{K}_+
\end{equation}
for a large enough $C_{\mathrm{kin}}>0$. On another hand, for the low kinetic energy part, we estimate the modified Hartree functional $\cE_{a,b,N}^{\mathrm{mH}}(\gamma)$ from below by the modified \NLS functional defined on mixed states, given by
\begin{equation}\label{functional:mixed-nls-inhomogeneous}
	\cE_{a,b}^{\mathrm{mNLS}}(\gamma) := \tr\left[ h\gamma \right] - \frac{a}{2} \int_{\R^2} \rho_{\gamma}(x)^2 \dix + \frac{b}{6} \int_{\R^2} \rho_{\gamma}(x)^3 \dix \,.
\end{equation}
For this purpose, we apply~\eqref{cv-ineq:hartree-nls-2body} and~\eqref{cv-ineq:hartree-nls-3body} with $v$ replaced by $\sqrt{\rho_{\gamma}}$, noticing that $\sqrt{\rho_{\gamma}} \in H^1(\R^2)$ for all $\gamma \in \mathrm{K}_-$. We obtain
\begin{equation}\label{cv:ground-states-split-2}
	\cE_{a,b,N}^{\mathrm{mH}}(\gamma) \geq \cE_{a,b}^{\mathrm{mNLS}}(\gamma) - o(1),\quad \forall \gamma \in \mathrm{K}_- \,.
\end{equation}
Putting together~\eqref{cv:ground-states-split-1} and~\eqref{cv:ground-states-split-2} into~\eqref{cv:ground-states-deF}, we deduce that
\begin{align*}
	o(1) + E_{a,b}^{\NLS} \geq E_{a,b,N}^{\mathrm{QM}} \geq{} & \int_{\mathrm{K}_+}\frac{\lambda C_{\mathrm{Kin}}}{2}\di\mu_{\Psi_N} (\gamma) + \int_{\mathrm{K}_-}\cE_{a,b}^{\mathrm{mNLS}}(\gamma)\di\mu_{\Psi_N} (\gamma) - o(1) \\
	\geq{} & \int_{\cS_{P}}\min\left\{\frac{\lambda C_{\mathrm{Kin}}}{2},\cE_{a,b}^{\mathrm{mNLS}}(\gamma)\right\}\di\mu_{\Psi_N} (\gamma) - o(1) \,.
\end{align*}
Recalling now ---see~\eqref{cv:energy_measure}--- that, modulo a subsequence, the sequence of measures~$\{\mu_{\Psi_N}\}_N$ converges to a measure~$\mu$ and 
passing to the limit $N\to+\infty$, we have
\[
	E_{a,b}^{\NLS} \geq \lim_{N\to+\infty} E_{a,b,N}^{\mathrm{QM}} \geq \int_{\cS} \min\left\{\frac{\lambda C_{\mathrm{Kin}}}{2},\cE_{a,b}^{\mathrm{mNLS}}(\gamma)\right\} \di\mu(\gamma) \,.
\]
Passing now to limit $C_{\mathrm{kin}}\to+\infty$ gives
\[
	E_{a,b}^{\NLS} \geq \lim_{N\to+\infty} E_{a,b,N}^{\mathrm{QM}} \geq \int_{\cS} \cE_{a,b}^{\mathrm{mNLS}}(\gamma) \di\mu(\gamma) \,.
\]
Moreover, testing~\eqref{ineq:information-deF} with $A$, $B$, and $C$ finite rank compact operators whose ranges lie within that of~$\widetilde{P} = \1_{h\leq K}$, using the convergences~$\mu_{\Psi_N} \to \mu$ and~$\gamma_{\Psi_N}^{(3)} \to \gamma^{(3)}$, then letting $K \to+\infty$, we obtain that this measure~$\mu$ is related to the measure~$\nu$ in~\eqref{eq:cv-quantum-de-finetti} through
\[
	\int_{\cS} \gamma^{\otimes 3} \di\mu(\gamma) = \int_{S\gH} |v^{\otimes 3} \rangle \langle v^{\otimes 3}| \di\nu(v) \,.
\]
Therefore, using this fact and that $\cE_{a, b}^{\mathrm{mNLS}}(\gamma)$ is a linear function of~$\gamma^{\otimes 3}$, we finally arrive at
\begin{equation}\label{cv:ground-states-last-step}
	E_{a,b}^{\NLS} \geq \lim_{N\to+\infty} E_{a,b,N}^{\mathrm{QM}} \geq \int_{\cS} \cE_{a, b}^{\mathrm{mNLS}}(\gamma) \di\mu(\gamma) = \int_{S\gH} \cE_{a, b}^{\NLS}(v) \di\nu (v) \geq E_{a,b}^{\NLS} \,.
\end{equation}
This shows that $\nu$ must also be supported on $\cM^{\NLS}$, which proves~\eqref{cv:qm-nls-ground-state}.

\medskip

We finish the proof of Theorem~\ref{thm:many-body} by investigating the collapse behavior of many\nobreakdash-body ground states in~\eqref{blowup:qm-ground-state}. Let $\{\Psi_N\}_N$ be a sequence of ground states of~$E_{a_N, b_N, N}^{\mathrm{QM}}$. It follows from~\eqref{ineq:information-deF} and the choice of~$L$ in~\eqref{eq:L} that we must have
\[
	\lim_{N\to+\infty}\tr\left| P^{\otimes 3} \gamma_{\Psi_N}^{(3)} P^{\otimes 3} - \int_{\cS_{P}} \gamma^{\otimes 3} \di\mu_{\Psi_N}(\gamma)\right| = 0 \,.
\]
The above together with~\eqref{ineq:measure-deF}, $\tr[\gamma_{\Psi_N}^{(3)}] = 1$, \eqref{cv:trace_Pperp_one-body}, and the triangle inequality yield
\[
	\lim_{N\to+\infty}\tr\left|\gamma_{\Psi_N}^{(3)} - \int_{\cS_{P}} \gamma^{\otimes 3} \di\mu_{\Psi_N}(\gamma)\right| = 0 \,.
\]
Taking a partial trace we also have
\[
	\lim_{N\to+\infty}\tr\left|\gamma_{\Psi_N}^{(1)} - \int_{\cS_{P}} \gamma \di\mu_{\Psi_N}(\gamma)\right| = 0 \,.
\]
To complete the proof of~\eqref{blowup:qm-ground-state}, it suffices to prove the convergence of $(\gamma_{\Psi_N}^{(1)} - |Q_N \rangle \langle Q_N|)$ to $0$ in trace class, where $Q_N = \ell_N^{-1}Q_0(\ell_N^{-1}\cdot)$. It follows from the above that this is equivalent to
\begin{equation}\label{blowup:ground-states-last-step}
	\lim_{N\to+\infty} \int_{\cS_{P}} \left|\langle \sqrt{\rho_{\gamma}}, Q_N \rangle\right| \di\mu_{\Psi_N}(\gamma) = 1 \,.
\end{equation}

To this end, we define
\[
	\delta_N := \int_{\cS_{P}} \frac{ \cE_{a_N, b_N, N}^{\mathrm{H}}(\sqrt{\rho_{\gamma}}) - E_{a_N, b_N, N}^{\mathrm{H}} }{ \left| E_{a_N, b_N, N}^{\mathrm{H}} \right| } \di\mu_{\Psi_N}(\gamma) \geq 0 \,,
\]
by definition of~$E_{a_N, b_N, N}^{\mathrm{H}}$, and claim that $\delta_N \to 0$. Indeed, on the one hand by the lower bound in~\eqref{ineq:measure-deF} together with~\eqref{cv:trace_Pperp_one-body} if $E_{a_N, b_N, N}^{\mathrm{H}}>0$ and on the other hand by the upper bound in~\eqref{ineq:measure-deF} if $E_{a_N, b_N, N}^{\mathrm{H}}<0$ ---which happens if $\zeta > 2 + 4/s$, see~\eqref{blowup:hartree-inhomogeneous-energy}---, we have
\[
	E_{a_N, b_N, N}^{\mathrm{H}} \leq (1+o(1)) E_{a_N, b_N, N}^{\mathrm{H}} \int_{\cS_{P}} \di\mu_{\Psi_N}(\gamma) \,.
\]
By~\eqref{cv:ground-states-deF}, this yields
\[
	(1+o(1)) E_{a_N, b_N, N}^{\mathrm{H}} \int_{\cS_{P}} \di\mu_{\Psi_N}(\gamma) \geq E_{a_N, b_N, N}^{\mathrm{H}} \geq \int_{\cS_{P}}\cE_{a_N, b_N, N}^{\mathrm{H}}(\sqrt{\rho_\gamma})\di\mu_{\Psi_N} (\gamma) - o(1)\,,
\]
hence
\[
	o(1) \sgn\left( E_{a_N, b_N, N}^{\mathrm{H}} \right) \int_{\cS_{P}} \di\mu_{\Psi_N}(\gamma) + \frac{o(1)}{ \left| E_{a_N, b_N, N}^{\mathrm{H}} \right| } \geq \delta_N \geq0
\]
and we obtain the claim, using~\eqref{cv:energy_measure} as well as the fact that the error term from~\eqref{cv:ground-states-deF} ---the one divided by $|E_{a_N, b_N, N}^{\mathrm{H}}|$ in the above equation--- is of small order compared to the \NLS energy (see the end of the previous subsection), hence compared to $E_{a_N, b_N, N}^{\mathrm{H}}$.

We also define $\cT_N$ as the set of all trace-class, self-adjoint operators $0\leq \gamma \in \cS_{P}$ s.t.\ $\tr\gamma = 1$ and
\begin{equation}\label{blowup:approximate-ground-states-last-step}
	0\leq \frac{ \cE_{a_N, b_N, N}^{\mathrm{H}}(\sqrt{\rho_{\gamma}}) - E_{a_N, b_N, N}^{\mathrm{H}} }{ \left| E_{a_N, b_N, N}^{\mathrm{H}} \right| } \leq \sqrt{\delta_N} \,.
\end{equation}
Note that the sets $\cT_N$ are non-empty since they contain Hartree ground states for instance.

On the one hand, we claim that we must have
\begin{equation}\label{limit:PsiQ_critical}
	\lim_{N\to+\infty}\inf_{\gamma\in \cT_N} \left|\langle \sqrt{\rho_{\gamma}}, Q_N \rangle\right| = 1 \,.
\end{equation}
Indeed, if this was not the case, and since $|\langle \sqrt{\rho_{\gamma}}, Q_N \rangle|\leq1$ by Schwarz inequality, there would exist a subsequence $\{\gamma_N \}_N \subset \cT_N$ such that
\begin{equation}\label{limit:PsiQ_fraud_critical}
	\limsup_{N\to+\infty} \left|\langle \sqrt{\rho_{\gamma_N}}, Q_N \rangle\right| < 1 \,.
\end{equation}
Since $\gamma_N \in \cT_N$ and $\delta_N \to 0$, we would deduce from~\eqref{blowup:approximate-ground-states-last-step} that
\[
	\lim_{N\to+\infty} \frac{ \cE_{a_N, b_N, N}^{\mathrm{H}}(\sqrt{\rho_{\gamma_N}}) }{E_{a_N, b_N, N}^{\mathrm{H}}} = 1 \,.
\]
Therefore, $\{\sqrt{\rho_{\gamma_N}}\}_N$ would be a sequence of approximate ground states of~$E_{a_N, b_N, N}^{\mathrm{H}}$ and
\[
	\lim_{N\to+\infty} \left|\langle \sqrt{\rho_{\gamma_N}}, Q_N \rangle\right| = 1 \,,
\]
by Theorem~\ref{thm:collapse-hartree}, contradicting~\eqref{limit:PsiQ_fraud_critical}. Thus,~\eqref{limit:PsiQ_critical} is proved.

On the other hand, by the definition of~$\delta_N$ and the choice of~$\cT_N$, we have
\[
	\delta_N \geq \int_{\cT_N^c} \frac{ \cE_{a_N, b_N, N}^{\mathrm{H}}(\sqrt{\rho_{\gamma}}) - E_{a_N, b_N, N}^{\mathrm{H}} }{ \left| E_{a_N, b_N, N}^{\mathrm{H}} \right| } \di\mu_{\Psi_N}(\gamma) \geq \sqrt{\delta_N} \, \mu_{\Psi_N}(\cT_N^c) \,,
\]
where $\cT_N^c = \cS_{P} \setminus \cT_N$.
Thus, $\mu_{\Psi_N}(\cT_N^c) \leq \sqrt{\delta_N} \to 0$ and consequently $\mu_{\Psi_N}(\cT_N) \to 1$, by~\eqref{cv:energy_measure}. The latter convergence and~\eqref{limit:PsiQ_critical} imply that
\begin{align*}
	\int_{\cS_{P}} \left|\langle \sqrt{\rho_{\gamma}}, Q_N \rangle\right| \di\mu_{\Psi_N}(\gamma) \geq{} & \int_{\cT_N} \left|\langle \sqrt{\rho_{\gamma}}, Q_N \rangle\right| \di\mu_{\Psi_N}(\gamma) \\
	\geq{} & \mu_{\Psi_N}(\cT_N)\inf_{\gamma \in \cT_N} \left|\langle \sqrt{\rho_{\gamma}}, Q_N \rangle\right| \xrightarrow[N \to+\infty]{} 1 \,.
\end{align*}
Thus~\eqref{blowup:ground-states-last-step} holds true and the proof of Theorem~\ref{thm:many-body} is completed.

\medskip

\begin{remark}\label{rem:bec-refined}
	There is room for improvement in our estimates and we can for instance obtain~\eqref{cv:qm-nls-energy} and~\eqref{cv:qm-nls-ground-state} for a wider range $0 < \alpha < 1/4$ and $0 < \beta < 1/8$. Let us point out important steps.
	
	Coming back to~\eqref{lower-bound:many-body} and using this time the first inequality in~\eqref{lower-bound:two-body-operator}, we obtain
	\begin{equation}\label{lower-bound:two-body-refined}
		\tr\left[\cU_{N^\alpha}\gamma_{\Psi_N}^{(3)}\right] \leq (1+\varepsilon)\tr\left[\cU_{N^\alpha}P^{\otimes 3}\gamma_{\Psi_N}^{(3)}P^{\otimes 3}\right] + C(1+\varepsilon^{-1})N^{2\alpha}\tr\left[ P_{\perp} \gamma_{\Psi_N}^{(1)} \right].
	\end{equation}
	Similarly, for the repulsive, three\nobreakdash-body interaction, we have the refined estimate
	\begin{align}\label{lower-bound:three-body-refined}
		\tr\left[W_{N^\beta}\gamma_{\Psi_N}^{(3)}\right] &\geq (1-\varepsilon)\tr\left[W_{N^\beta}P^{\otimes 3}\gamma_{\Psi_N}^{(3)}P^{\otimes 3}\right] - C \varepsilon^{-1} N^{4\beta}\tr\left[ P_{\perp} \gamma_{\Psi_N}^{(1)} \right] \nonumber\\
		&\geq (1-\varepsilon)\tr\left[W_{N^\beta}P^{\otimes 3}\gamma_{\Psi_N}^{(3)}P^{\otimes 3}\right] - C(1 + \varepsilon^{-1}) N^{4\beta}\tr\left[ P_{\perp} \gamma_{\Psi_N}^{(1)} \right].
	\end{align}
	It then follows from~\eqref{lower-bound:one-body},~\eqref{lower-bound:two-body-refined}, and~\eqref{lower-bound:three-body-refined} that
	\begin{multline}\label{lower-bound:many-body-refined}
		\frac{\langle \Psi_N, H_{a,b,N} \Psi_N \rangle}{N} \geq \frac{1}{3}\tr\left[H_{a, b, N, \varepsilon}^{(3)}P^{\otimes 3}\gamma_{\Psi_N}^{(3)}P^{\otimes 3}\right] \\
		+ \left(L - C \left( 1+\varepsilon^{-1} \right) \left(N^{2\alpha}  + N^{4\beta} \right) \right)\tr\left[ P_{\perp} \gamma_{\Psi_N}^{(1)} \right]
	\end{multline}
	where $H_{a, b, N, \varepsilon}^{(3)}$ is $H_{a,b,N}^{(3)}$ in~\eqref{lower-bound:many-body} but with $a$ and $b$ replaced by $(1+\varepsilon)a$ and $(1-\varepsilon)b$, respectively. We might apply again Theorem~\ref{thm:information_deF} to obtain that
	\begin{equation}\label{lower-bound:deF-refined}
		\frac{1}{3} \tr\left[H_{a, b, N, \varepsilon}^{(3)}P^{\otimes3}\gamma_{\Psi_N}^{(3)}P^{\otimes3} \right] \geq \int_{\cS_{P}} \cE_{a,b,N,\varepsilon}^{\mathrm{mH}}(\gamma)\di\mu_{\Psi_N}(\gamma) - C \sqrt{\frac{\log L}{N}} \left( L + N^{2\alpha}+N^{4\beta}\right)
	\end{equation}
	where $\cE_{a,b,N,\varepsilon}^{\mathrm{mH}}$ is $\cE_{a,b,N}^{\mathrm{mH}}$ in~\eqref{lower-bound:main} but with $a$ and $b$ replaced by $(1+\varepsilon)a$ and $(1-\varepsilon)b$, respectively. Note that the associated Hartree energy $E_{a,b,N,\varepsilon}^{\mathrm{mH}}$ is bounded below uniformly in $N$ when $\varepsilon$ is chosen independent of~$N$ and such that $0<\varepsilon < a_*/a - 1$ if $a < a_*$ and $0<\varepsilon<1$ if $a\geq a_*$ and $\alpha<\beta$ (this choice is again made in order to ensure that $(1+\varepsilon)a$ compares the same way to $a_*$ as $a$ does). This can be obtained by the same arguments as in the proof of Theorem~\ref{thm:collapse-hartree}. Choosing now $L$ as a power of $N$ such that $N^{\frac{1}{2}} \gg L \gg N^{2\alpha} + N^{4\beta}$ ---hence $\alpha < 1/4$ and $\beta < 1/8$---, as well as
	\[
		\varepsilon = \tilde{C} (N^{2\alpha}+N^{4\beta})L^{-1}
	\]
	for a sufficient large $\tilde{C}>0$, and using~\eqref{lower-bound:many-body-refined} and~\eqref{lower-bound:deF-refined}, we obtain
	\begin{equation}\label{cv:ground-states-deF-refined}
		\frac{ \langle \Psi_N, H_{a,b,N} \Psi_N \rangle }{N} \geq \int_{\cS_{P}} \cE_{a,b,N,\varepsilon}^{\mathrm{mH}}(\gamma)\di\mu_{\Psi_N}(\gamma) + \frac{L}{2}\tr\left[ P_{\perp} \gamma_{\Psi_N}^{(1)} \right] - C \sqrt{\frac{\log L}{N}} L
	\end{equation}
	and the choice $L=N^\delta$ with $\delta<1/2$ that we just made ensures the energy convergence. For the convergence of ground states, we process as in the arguments from~\eqref{cv:ground-states-deF} to~\eqref{cv:ground-states-last-step}, using~\eqref{cv:ground-states-deF-refined}, to obtain the following refined estimate of~\eqref{cv:ground-states-last-step}
	\[
		E_{a,b}^{\NLS} \geq \lim_{N\to+\infty} E_{a,b,N}^{\mathrm{QM}} \geq \int_{\cS} \cE_{a,b,\varepsilon}^{\mathrm{mNLS}}(\gamma) \di\mu(\gamma) = \int_{S\gH} \cE_{a,b,\varepsilon}^{\NLS}(v) \di\nu (v) \geq E_{a,b,\varepsilon}^{\NLS} \,.
	\]
	Here $\cE_{a,b,\varepsilon}^{\mathrm{mNLS}}$, $\cE_{a,b,\varepsilon}^{\NLS}$, and $E_{a,b,\varepsilon}^{\NLS}$ are respectively $\cE_{a,b}^{\mathrm{mNLS}}$ in~\eqref{functional:mixed-nls-inhomogeneous}, $\cE_{a,b}^{\NLS}$ in~\eqref{functional:nls-inhomogeneous}, and $E_{a,b}^{\NLS}$ in~\eqref{energy:nls-inhomogeneous} but with $a$ and $b$ replaced by $(1+\varepsilon)a$ and $(1-\varepsilon)b$, respectively. Passing finally to the limit $\varepsilon \to 0$, we obtain~\eqref{cv:qm-nls-ground-state} as well as~\eqref{cv:qm-nls-energy} for $\alpha < 1/4$ and $\beta < 1/8$.
\end{remark}

\appendix
\section{Homogeneous BECs}\label{app:homogeneous}

In this appendix, we consider the auxiliary minimization problem without external potential, i.e., without the term $\sum_{i=1}^N |x_i|^s$ in~\eqref{hamiltonian}. We are interested in ground states of the following homogeneous \NLS minimization problem
\begin{equation}\label{energy:nls-homogeneous}
	G^{\NLS}_{a,b} := \inf\left\{\cG^{\NLS}_{a,b}(v):v\in H^1(\R^2), \norm{v}_{L^2}=1\right\},
\end{equation}
with the associated homogeneous \NLS functional
\begin{equation}\label{functional:nls-homogeneous}
	\cG^{\NLS}_{a,b}(v) := \norm{ \nabla v }_{L^2}^2 - \frac{a}{2} \norm{v}_{L^4}^4 + \frac{b}{6} \norm{v}_{L^6}^6 \,.
\end{equation}
In contrast to the inhomogeneous case~\eqref{energy:nls-inhomogeneous}, the homogeneous minimization problem~\eqref{energy:nls-homogeneous} with nontrivial, three\nobreakdash-body interaction only admits a ground state when the two\nobreakdash-body interaction is sufficiently negative: $a>a_*$. Furthermore, by establishing the asymptotic behaviors of~\eqref{energy:nls-homogeneous} and of its ground states, in the case where $a_n$ converges from above ($a_n \searrow a_*$) with a convergence rate slower in order of magnitude than the one of $b_n \searrow 0$ (see Remark~\ref{rem:collapse-nls-inhomogeneous}), and comparing it with~\eqref{energy:nls-inhomogeneous}, we also prove the missing case~$\zeta = +\infty$ in Theorems~\ref{thm:collapse-nls} and~\ref{thm:many-body}. We have the following.

\begin{theorem}[Homogeneous \NLS ground states: existence and collapse]\label{thm:collapse-nls-homogeneous}
	Let $a,b\in\R$ and $G^{\NLS}_{a,b}$ be given in~\eqref{energy:nls-homogeneous}.
	\begin{enumerate}[label=(\roman*)]
		\item\label{existence:nls-homogeneous-gound-state} $G^{\NLS}_{a,b}$ admits a ground state if and only if either ($b=0$ and $a=a_*$) or ($b>0$ and $a>a_*$).
		\item\label{thm-collapse-nls_homogeneous} Let $\{a_n\}_n \subset(a_*, +\infty)$ and $\{b_n\}_n \subset (0,+\infty)$ satisfy $a_n \searrow a_*$ and $b_n \searrow 0$ as $n\to+\infty$.
		Let $\{v_n\}_n$ be a sequence of (approximate) ground states of~$G_{a_n, b_n}^{\NLS}$. Then, up to a translation,
		\begin{equation}\label{thm:collapse-nls-homogeneous-ground-state}
			\lim_{n\to+\infty} \ell_n v_n(\ell_n \cdot) = Q_0
		\end{equation}
		strongly in~$H^1(\R^2)$, for the whole sequence, where
		\begin{equation}\label{thm:collapse-nls-homogeneous-blowup-length}
			\ell_n = \sqrt{2a_* \cQ_6 \frac{b_n}{a_n - a_*}} \,,
		\end{equation}
		with $\cQ_6$ and $Q_0$ are defined in~\eqref{Def-Q-s} and~\eqref{Def_Q0}, respectively. Furthermore,
		\begin{equation}\label{thm:collapse-nls-homogeneous-energy}
			G_{a_n, b_n}^{\NLS} = -\frac{(a_n - a_*)^2}{4 a_*^2 \cQ_6 b_n} (1+o(1)) \,.
		\end{equation}
	\end{enumerate}
\end{theorem}

\begin{remark}\label{rem:collapse-nls-homogeneous}\leavevmode
	\begin{itemize}[leftmargin=*]
		\item For~\emph{\ref{existence:nls-homogeneous-gound-state}}, we have in the critical case where $a=a_*$ and $b=0$ that $G^{\NLS}_{a_*, 0} = 0$ with the ground states being the normalized optimizers of~\eqref{ineq:GN}, and in the case $a>a_*$ and $b>0$ that $G^{\NLS}_{a,b}<0$ with its ground states that can be chosen nonnegative, \emph{radially symmetric decreasing}, by~\cite[Theorems 7.8 and 7.17]{LieLos-01}.
		
		\item For~\emph{\ref{thm-collapse-nls_homogeneous}}, that is, the collapse regime where $a_n \searrow a_*$ and $b_n \searrow 0$, we define $A_n:=(a_n - a_*)/a_*>0$.
		
			On the one hand if $\ell_n$ in~\eqref{thm:collapse-nls-homogeneous-blowup-length} satisfies $\lim_{n\to+\infty} \ell_n \in (0,+\infty]$ ---which implies $\zeta=1$ in~\eqref{thm-collapse-nls-condition-ratio-sequences}---, then the homogeneous energy~$G_{a_n, b_n}^{\NLS}$ of course converges to zero, by~\eqref{thm:collapse-nls-homogeneous-energy}, but its (approximate) ground states do not collapse, while comparatively the ground states of the inhomogeneous energy~\eqref{energy:nls-inhomogeneous} collapse to~$Q_0$, at the blow-up length scale~\eqref{thm-collapse-nls-condition-ratio-sequences}, as shown in Theorem~\ref{thm:collapse-nls}.
			
			On the other hand, if $\lim_{n\to+\infty} \ell_n = 0$, i.e., $b_n=o(A_n)$, then the homogeneous energy can still convergence ---as long as $\lim_{n\to+\infty} A_n^2/b_n < +\infty$---, including to zero ---if $A_n^2=o(b_n)$\textemdash, even though the $H^1(\R^2)$-norms of~$v_n$ diverge since $\norm{\nabla v_n}_{L^2}^2 = \ell_n^{-2} \norm{ \nabla Q_0 }_{L^2}^2 = \ell_n^{-2}$ diverges. It is worth noting that this case $b_n \ll A_n$ includes the missing case~$\zeta = +\infty$ in Theorem~\ref{thm:collapse-nls}, which reads $b_n^{(s+2)/(s+4)} \ll A_n$. We claim that~\eqref{thm:collapse-nls-homogeneous-ground-state} and~\eqref{thm:collapse-nls-homogeneous-energy} then fill the gap of this missing case in Theorem~\ref{thm:collapse-nls}. Indeed, if $a_n \searrow a_*$ and $b_n \searrow 0$ in such a way that $b_n^{(s+2)/(s+4)} \ll A_n$, then as just explained $\ell_n$ in~\eqref{thm:collapse-nls-homogeneous-blowup-length} converges to $0$ as $n\to+\infty$. Looking back at~\eqref{nls-energy-upperbound}, one observes that the term due to the trapping potential vanishes, while the two other terms explode, hence it has in particular no contribution to the leading order in this collapse regime. By comparing $E^{\NLS}_{a_n, b_n}$ and $G^{\NLS}_{a_n, b_n}$, we obtain that
			\[
				E_{a_n, b_n}^{\NLS} = G_{a_n, b_n}^{\NLS} (1+o(1)) = -\frac{(a_n - a_*)^2}{4 a_*^2 \cQ_6 b_n} (1+o(1))
			\]
			and the sequence of (approximate) ground states of~$E_{a_n, b_n}^{\NLS}$ blows up to $Q_0$ at the blow-up length~$\ell_n$ in~\eqref{thm:collapse-nls-homogeneous-blowup-length}, as shown in~\eqref{thm:collapse-nls-homogeneous-ground-state}.
		
		\item In contrast to the inhomogeneous case, where the blow-up phenomenon does not occur in the limit $a_n \searrow a_*$ with $b>0$ being fixed (see Remark~\ref{rem:collapse-nls-inhomogeneous}), the sequence of ground states $\{v_n\}_n$ of~$G_{a_n, b}^{\NLS}$ blows up in that regime, in the sense that $\norm{\nabla v_n}_{L^2} \to+\infty$. This is due to the fact that $G_{a_*, b}^{\NLS}$ does not admit a ground state. In such collapse regime, which is not covered by the statement of Theorem~\ref{thm:collapse-nls-homogeneous}, the asymptotic behaviors of~$G_{a_n, b}^{\NLS}$ and of its ground states $\{v_n\}_n$ are given by~\eqref{thm:collapse-nls-homogeneous-energy} and~\eqref{thm:collapse-nls-homogeneous-ground-state}.
		
		\item The blow-up phenomenon is also observed in the regime where $a>a_*$ is fixed and $b_n \searrow 0$ (also not covered by the statement of Theorem~\ref{thm:collapse-nls-homogeneous}). It does not occur in the sense $\norm{\nabla v_n}_{L^2} \to+\infty$, but in the sense $G_{a,b_n}^{\NLS} \to -\infty$, as given by~\eqref{thm:collapse-nls-homogeneous-energy}. For this reason, a sequence of \emph{approximate} \NLS ground states~$\{v_n\}_n$ is now naturally defined through
			\[
				G_{a,b_n}^{\NLS} \leq \cG_{a,b_n}^{\NLS}(v_n) \leq G_{a,b_n}^{\NLS} + o(1) \,.
			\]
			instead of~\eqref{nls:inhomogeneous-approximate-ground-state}. Furthermore, in this collapse regime, and for the same reason as two items above, the term in~\eqref{nls-energy-upperbound} due to the trapping potential has no contribution to the leading order term since $\ell_n \to 0$. The asymptotic behaviors of~$E_{a,b_n}^{\NLS}$ and $G_{a,b_n}^{\NLS}$ and of their ground states are then essentially the same. It is worth noting that the blowup profile of the sequence of ground states is not given by the solution of the cubic \NLS equation~\eqref{eq:NLS} but by the ground states of~$G_{a,1}^{\NLS}$ (see Theorem~\ref{thm:collapse-hartree-homogeneous} for further discussions).
	\end{itemize}
\end{remark}

\begin{proof}[Proof of Theorem~\ref{thm:collapse-nls-homogeneous}]
	We start by proving a priori bounds on~$G^{\NLS}_{a,b}$ for $(a,b)\in\R^2$. On the one hand, by the variational principle,
	\begin{equation}\label{upper-bound:nls-homogeneous}
		G^{\NLS}_{a,b} \leq \cE_{a,b}^{\NLS}\left(\ell Q(\ell \cdot)\right) = \ell^2 \frac{a_* - a}{2} \norm{Q}_{L^4}^4 + \ell^4 \frac{b}{6} \norm{Q}_{L^6}^6
	\end{equation}
	for any $L^2$-normalized optimizer $Q$ of~\eqref{ineq:GN}. The above yields that $G^{\NLS}_{a,b} \leq0$ for any $(a,b)\in\R^2$, by taking $\ell \to 0$.
	Moreover, taking $\ell \to+\infty$ in~\eqref{upper-bound:nls-homogeneous}, we obtain that $G^{\NLS}_{a,b} = -\infty$ when either ($b<0$ and $a\in\R$) or ($b=0$ and $a>a_*$). On the other hand, by~\eqref{ineq:GN}, we have
	\begin{equation}\label{lower-bound:nls-homogeneous}
		\cG_{a,b}^{\NLS}(v) \geq \frac{a_* - a}{2} \norm{v}_{L^4}^4 + \frac{b}{6} \norm{v}_{L^6}^6 \,,
	\end{equation}
	for any $L^2$-normalized $v\in H^1(\R^2)$. When $b\geq0$ and $a\leq a_*$, the above yields that $G^{\NLS}_{a,b} \geq 0$, hence $G^{\NLS}_{a,b} = 0$. In the special case $a=a_*$ and $b=0$, we of course have $\cE_{a_*, 0}^{\NLS}(Q) = 0$ for the unique (up to translation and dilation) normalized optimizer $Q$ of~\eqref{ineq:GN}, which is therefore a ground state. Otherwise, when $b\geq0$ and $a\leq a_*$ but $(a,b)\neq(a_*, 0)$, then the right hand side of~\eqref{lower-bound:nls-homogeneous} is strictly positive, in particular if $v$ is a ground state of~$G^{\NLS}_{a,b}$. This is a contradiction to $G^{\NLS}_{a,b} \leq0$ for any $(a,b)\in\R^2$ and yields that there does not exit ground states when either ($a <a_*$ and $b=0$) or ($a\leq a_*$ and $b>0$).
	
	To conclude the proof of~\emph{\ref{existence:nls-homogeneous-gound-state}}, we now deal with the most demanding case where $b>0$ and $a>a_*$. We first prove the existence of a ground state. Let $\{v_n\}_n$ be a \emph{radially symmetric decreasing} minimizing sequence for $G^{\NLS}_{a,b}$. We thus have
	\[
		G^{\NLS}_{a,b} = \lim_{n\to+\infty}\cG_{a,b}^{\NLS}(v_n) \quad \text{with} \quad \norm{v_n}_{L^2}=1 \,.
	\]
	Replacing $Q$ by $Q_0$ in~\eqref{upper-bound:nls-homogeneous} and optimizing over $\ell>0$, we get
	\begin{equation}\label{upper-bound:hartree-homogeneous}
		G_{a,b}^{\NLS} \leq -\frac{(a - a_*)^2}{4 a_*^2 \cQ_6 b} < 0 \,.
	\end{equation}
	Recalling that the cubic term is controlled by the quintic term, see~\eqref{ineq:2-to-3}, equation~\eqref{upper-bound:hartree-homogeneous} gives that $\norm{\nabla v_n}_{L^2}$ is uniformly bounded. Thus, there exists $v\in H^1(\R^2)$ such that, up to a translation and a subsequence, $v_n \to v$ weakly in~$H^1(\R^2)$ and, by the compact embedding $H_{\text{rad}}^1(\R^2) \hookrightarrow L^p(\R^2)$ for $2<p<+\infty$ (see, e.g.,~\cite{Strauss-77}), strongly in~$L^q(\R^2)$ for any $2< q<+\infty$. Fatou's lemma then yields
	\[
		0 > G^{\NLS}_{a,b} = \lim_{n\to+\infty}\cG_{a,b}^{\NLS}(v_n) \geq \cG_{a,b}^{\NLS}(v) \geq \norm{ v }_{L^2}^2 \cG_{a,b}^{\NLS}(v(\norm{ v }_{L^2}\cdot)) \geq G^{\NLS}_{a,b} \,,
	\]	
	since the first two inequalities give $v\not\equiv0$, then we use $0 < \norm{ v }_{L^2} \leq 1$ to obtain the third one, and use it again with $G^{\NLS}_{a,b} < 0$ for the last one. Therefore, we must have equalities in the above. In particular, we have $\norm{ v }_{L^2} = 1$ and $v$ is a ground state of~$G^{\NLS}_{a,b}$.
	
	\medskip
	
	Next, we prove~\emph{\ref{thm-collapse-nls_homogeneous}}. That is, the asymptotic behaviors of the homogeneous \NLS energy and ground states. The energy upper bound in~\eqref{thm:collapse-nls-homogeneous-energy} follows from~\eqref{upper-bound:hartree-homogeneous}. We prove the matching energy lower bound by proving~\eqref{thm:collapse-nls-homogeneous-ground-state}. Let $\{v_n\}_n$ be a sequence of (approximate) ground states of~$G_{a_n, b_n}^{\NLS}$ and let $w_n := \ell_n v_n(\ell_n \cdot)$. Then, $\norm{w_n}_{L^2} = \norm{v_n}_{L^2}=1$ and
	\begin{equation}\label{energy-blowup:nls-homogeneous}
		G^{\NLS}_{a_n, b_n} = \cG^{\NLS}_{a_n, b_n} (v_n) = \ell_n^{-2}\left(\norm{ \nabla w_n }_{L^2}^2 - \frac{a_n}{2} \norm{w_n}_{L^4}^4 \right) + \ell_n^{-4}\frac{b_n}{6}\norm{w_n}_{L^6}^6 \,.
	\end{equation}
	Multiplying~\eqref{energy-blowup:nls-homogeneous} by $2a_* \ell_n^2(a_n - a_*)^{-1}$ and using~\eqref{ineq:GN} as well as the energy lower bound just proved, we obtain
	\begin{equation}\label{energy-blowup:nls-homogeneous-nonvanishing}
		-1 \geq -a_* \norm{w_n}_{L^4}^4 + \frac{1}{6 \cQ_6} \norm{w_n}_{L^6}^6 \,.
	\end{equation}
	Since $\norm{w_n}_{L^4}^4 \leq \norm{w_n}_{L^6}^{3}$, the above yields that both $\norm{w_n}_{L^4}$ and $\norm{w_n}_{L^6}$ are bounded uniformly from above and below. Multiplying~\eqref{energy-blowup:nls-homogeneous} by $\ell_n^2$ and using~\eqref{upper-bound:hartree-homogeneous}, we obtain
	\begin{equation}\label{energy-blowup:nls-homogeneous-nondichotomy}
		0 \geq \lim_{n\to+\infty}\norm{ \nabla w_n }_{L^2}^2 - \frac{a_n}{2} \norm{w_n}_{L^4}^4 = \lim_{n\to+\infty} \norm{ \nabla w_n }_{L^2}^2 - \frac{a_*}{2} \norm{w_n}_{L^4}^4 \geq 0
	\end{equation}
	and the two inequalities are therefore actually equalities. In~\eqref{energy-blowup:nls-homogeneous-nondichotomy}, the first inequality comes from the nonnegativity of the quintic term, the equality from the boundedness of~$\norm{w_n}_{L^4}$ and the convergence of $a_n$, hence at this step we obtain the uniform boundedness of $\{w_n\}_n$ in~$H^1(\R^2)$, and the last inequality from~\eqref{ineq:GN}.
	
	Now, since $\{w_n\}_n$ is bounded uniformly in~$H^1(\R^2)$, there exists $w\in H^1(\R^2)$ such that, up to a translation and a subsequence, $w_n \to w$ weakly in~$H^1(\R^2)$ and almost everywhere in~$\R^2$. We prove that the convergence is actually strong by showing $\norm{w}_{L^2}=1$. If $\norm{w}_{L^2} = 0$, then $w_n \to 0$ strongly in~$L^p(\R^2)$ for all $2<p<\infty$ (see, e.g.,~\cite[Lemma 9]{LenLew-11}). Taking the limit $n\to+\infty$ in~\eqref{energy-blowup:nls-homogeneous-nonvanishing} we obtain a contradiction. Therefore, $w \not\equiv 0$. Next, we discard the possibility that $\norm{w}_{L^2} \in (0,1)$. If this was the case, we would deduce from~\eqref{energy-blowup:nls-homogeneous-nondichotomy} and from the Brezis--Lieb's refinement of Fatou's lemma (see, e.g.,~\cite[Theorem 1.9]{LieLos-01}) that
	\[
		0 = \norm{ \nabla w }_{L^2}^2 - \frac{a_*}{2} \norm{w}_{L^4}^4 + \lim_{n\to+\infty} \norm{ \nabla ( w_n - w ) }_{L^2}^2 - \frac{a_*}{2} \norm{ w_n - w }_{L^4}^4 \,.
	\]
	However, this is not possible since on the one hand
	\[
		\norm{ \nabla w }_{L^2}^2 - \frac{a_*}{2} \norm{w}_{L^4}^4 > 0
	\]
	by~\eqref{ineq:GN} together with $\norm{w}_{L^2} < 1$ and on the other hand, again by~\eqref{ineq:GN},
	\[
		\lim_{n\to+\infty}\norm{ \nabla ( w_n - w ) }_{L^2}^2 - \frac{a_*}{2} \norm{ w_n - w }_{L^4}^4 \geq 0\,.
	\]
	Therefore, we must have $\norm{w}_{L^2}=1$ and $w_n \to w$ strongly in~$L^p(\R^2)$ for all $2\leq p <\infty$.
	
	Furthermore,~\eqref{energy-blowup:nls-homogeneous-nondichotomy} implies that $w$ is an optimizer of~\eqref{ineq:GN}. Hence, after a suitable rescaling, $w(x) = \sqrt{t}Q_0(\sqrt{t}x)$ for some $t>0$, up to translation. We now show that $w \equiv Q_0$ up to translation, which proves~\eqref{thm:collapse-nls-homogeneous-ground-state} hence~\eqref{thm:collapse-nls-homogeneous-energy}. That is, we claim that $t= 1$. Indeed, multiplying again~\eqref{energy-blowup:nls-homogeneous} by $2a_* \ell_n^2(a_n - a_*)^{-1}$, using again~\eqref{ineq:GN} and~\eqref{upper-bound:hartree-homogeneous}, and taking the limit $n\to+\infty$, we get
	\[
		-1 \geq -2 \norm{ \nabla w }_{L^2}^2 + \frac{1}{6 \cQ_6}\norm{w}_{L^6}^6 = -2t + t^2 \geq -1 \,.
	\]
	Consequently, equality holds in the above, which implies $t=1$.
\end{proof}

With Theorem~\ref{thm:collapse-nls-homogeneous}-\emph{\ref{thm-collapse-nls_homogeneous}} in hand, we are now able to fill the gap of the case $\zeta = +\infty$ in Theorem~\ref{thm:many-body}.

\begin{theorem}[Many\nobreakdash-body ground states: condensation and collapse]\label{thm:many-body-homogeneous}
	Let $\alpha, \beta> 0$. Assume $U$ and $W$ satisfy~\eqref{condition:potential-two-body}--\eqref{condition:potential-three-body-symmetry} and $|x|U(x) \in L^1(\R^2)$. Let $\{a_N\}_N \subset (a_*, +\infty)$ and $\{b_N\}_N \subset (0,+\infty)$ satisfy $a_N - a_* \sim N^{-\mu}$ and $b_N \sim N^{-\nu}$ with
	\begin{equation}\label{many-body-blowup:speed-homogeneous}
		0<2\mu\leq \nu \qquad \text{ and } \qquad 2\nu - \mu < \alpha < \beta < \frac{1}{16} + \frac{\nu - 2\mu}{4} \,.
	\end{equation}
	Let $\{\Psi_N \}_N$ be a sequence of ground states of~$E_{a_N, b_N, N}^{\mathrm{QM}}$ given by~\eqref{energy:quantum} and let $\Phi_N := \ell_N \Psi_N(\ell_N \cdot)$ with $\ell_N = \sqrt{2a_* \cQ_6 \frac{b_N}{a_N - a_*}}$. Then, up to a translation,
	\[
		\lim_{N\to+\infty} \tr \left| \gamma_{\Phi_N}^{(k)} - | Q_0^{\otimes k} \rangle \langle Q_0^{\otimes k} | \right| = 0\,,\quad \forall\, k=1,2,\dots
	\]
	for the whole sequence, where $Q_0$ is given by~\eqref{Def_Q0}. Furthermore,
	\[
		E_{a_N, b_N, N}^{\mathrm{QM}} = E_{a_N, b_N}^{\NLS} (1+o(1)) = G_{a_N, b_N}^{\NLS} (1+o(1)) = -\frac{(a_N - a_*)^2}{4 a_*^2 \cQ_6 b_N}(1+o(1)) \,.
	\]
\end{theorem}
The proof of Theorem~\ref{thm:many-body-homogeneous} being similar to the one of Theorem~\ref{thm:many-body}, we skip the details for brevity, but we explain in particular where the precise conditions on $\alpha$, $\beta$, $\nu$, and $\mu$ come from.
\begin{proof}
	First, the condition $\alpha<\beta$ comes from the assumption $a_N>a_*$, see the proof of Theorem~\ref{thm:collapse-hartree}.
	
	\medskip
	
	Second, the condition $0<2\mu\leq \nu$ in~\eqref{many-body-blowup:speed-homogeneous} ensures that $b_N \searrow 0$ faster than $a_N \searrow a_*$ in such a way that $E_{a_N, b_N}^{\NLS} \sim (a_N-a_*)^2/b_N \sim N^{\nu-2\mu} \not\to 0$ (see Remark~\ref{rem:collapse-nls-homogeneous}).
	
	\medskip
	
	Third, the lower bound $2\nu - \mu < \alpha$ comes from the convergence of the Hartree energy (which has been studied in Theorem~\ref{thm:collapse-hartree} in the general case). Indeed, the upper bound~\eqref{blow-up:hartree-upper-bound} adapted to our present case, in which $\ell_N^2 = 2a_* \cQ_6 b_N/(a_N - a_*) \sim N^{\mu - \nu} \to 0$, reads
	\[
		E_{a_N, b_N, N}^{\mathrm{H}} \leq \cE_{a_N, b_N}^{\NLS}(\ell_N^{-1}Q_0(\ell_N^{-1}\cdot)) + CN^{-\alpha} \ell_N^{-3} = -\frac{ N^{\nu-2\mu} }{ 4 \cQ_6 } + CN^{3(\nu-\mu)-\alpha} + \frac{\cQ_s}{s} \ell_N^s \,,
	\]
	where we used~\eqref{nls-energy-upperbound}. Since $\ell_N\to0$, the lower bound $2\nu - \mu < \alpha$ is the largest range of $\alpha$'s for which our method can give the convergence of the Hartree energy. Note that this is always a stronger condition than $\alpha>0$ since $2\nu - \mu > \nu - 2\mu \geq0$.

	The matching lower bound on $E_{a_N, b_N, N}^{\mathrm{H}}$ is still obtained by proving an $H^1(\R^2)$-convergence of rescaled ground states~$w_N =\ell_N v_N(\ell_N \cdot)$, which requires the assumption $0<\alpha<\beta$ to first prove their uniform boundedness in $L^2(\R^2)$. See the arguments starting at~\eqref{conv:hartree-nls-ground-state} in the proof of~\eqref{blowup:hartree-inhomogeneous-ground-state}.
	
	It is worth nothing that the external potential has no contribution to the leading order in the collapse regime, which could be a problem for compactness. Fortunately, because the limiting profile is unique, we can still obtain the compactness of~$\{w_N\}_N$. For this, we return to~\eqref{blow-up:hartree-lower-bound}. We now neglect there the nonnegative external term and adapt the proof of the case $\zeta\neq0$. Using the upper bound of the Hartree energy and the definition of~$\ell_N$, we obtain
	\begin{multline}\label{energy-blowup:hartree-homogeneous-nonvanishing}
		-1 + o(1) \geq \frac{1}{6 \cQ_6} \iiint_{\R^6} W_{N^\beta\ell_N}(x-y, x-z)|w_N(x)|^2|w_N(y)|^2|w_N(z)|^2 \dix \diy \diz \\
		+ \frac{2a_*}{a_N - a_*}\left(\norm{\nabla w_N}_{L^2}^2 - \frac{a_N}{2} \norm{w_N}_{L^4}^4 \right).
	\end{multline}
	Multiplying both sides by $(a_N - a_*)/(2a_*)>0$, taking the limit $N\to+\infty$, and neglecting the three\nobreakdash-body term gives
	\begin{equation}\label{energy-blowup:hartree-homogeneous-nondichotomy}
		0 \geq \lim_{N\to+\infty}\norm{\nabla w_N}_{L^2}^2 - \frac{a_N}{2}\norm{w_N}_{L^4}^4 = \lim_{N\to+\infty}\norm{\nabla w_N}_{L^2}^2 - \frac{a_*}{2}\norm{w_N}_{L^4}^4 \geq 0 \,,
	\end{equation}
	due to~\eqref{ineq:GN} for the last inequality and, for the equality, due to the same arguments as for in~\eqref{energy-blowup:nls-homogeneous-nondichotomy}. The above is actually an equality. On the other hand, since $\{w_N\}_N$ is bounded uniformly in~$H^1(\R^2)$, there exists $w\in H^1(\R^2)$ such that, up to translation and extraction of a subsequence, $w_N \to w$ weakly in~$H^1(\R^2)$ and almost everywhere in~$\R^2$. We prove that this is actually strong convergence by showing $\norm{w}_{L^2}=1$. If $\norm{w}_{L^2} = 0$, then we must have that $w_N \to 0$ strongly in~$L^p(\R^2)$ for all $2<p<\infty$ (see, e.g.,~\cite[Lemma 9]{LenLew-11}). Hence, by~\eqref{cv-ineq:hartree-nls-3body} and the fact that $N^{\beta}\ell_N \sim N^{\beta - (\nu - \mu)/2} \to+\infty$ (since we already assumed $\beta > \alpha > 2\nu - \mu$), we have
	\begin{equation}\label{energy-blowup:hartree-homogeneous-nonvanishing-1}
		\lim_{N\to+\infty}\frac{1}{6 \cQ_6} \iiint_{\R^6} W_{N^\beta\ell_N}(x-y, x-z)|w_N(x)|^2|w_N(y)|^2|w_N(z)|^2 \dix \diy \diz = 0 \,.
	\end{equation}
	Moreover, again by ~\eqref{ineq:GN},
	\begin{equation}\label{energy-blowup:hartree-homogeneous-nonvanishing-2}
		\lim_{N\to+\infty}\frac{2a_*}{a_N - a_*}\left(\norm{\nabla w_N}_{L^2}^2 - \frac{a_N}{2}\norm{w_N}_{L^4}^4 \right) \geq \lim_{N\to+\infty} -a_* \norm{w_N}_{L^4}^4 = 0 \,.
	\end{equation}
	Taking the limit $N\to+\infty$ in~\eqref{energy-blowup:hartree-homogeneous-nonvanishing} contradicts that~\eqref{energy-blowup:hartree-homogeneous-nonvanishing-1} and~\eqref{energy-blowup:hartree-homogeneous-nonvanishing-2} can hold simultaneously. Therefore, $w \not\equiv 0$. On the other hand, by using~\eqref{energy-blowup:hartree-homogeneous-nondichotomy} and the same arguments as in the proof of Theorem~\ref{thm:collapse-nls-homogeneous}, we discard the possibility that $\norm{w}_{L^2} \in (0,1)$. Therefore, we must have $\norm{w}_{L^2}=1$ and $w_N \to w$ strongly in~$L^p(\R^2)$ for all $2\leq p <\infty$. Again, by the same arguments as in the proof of Theorem~\ref{thm:collapse-nls-homogeneous}, we can actually prove the energy lower bound for the Hartree energy as well as the $H^1(\R^2)$-convergence of~$w_N$ to $w \equiv Q_0$.
	
	\medskip
	
	 Fourth, and finally, the upper bound $\beta < 1/16 + (\nu - 2\mu)/4$ comes from the convergence of the quantum energy. Indeed, omitting the nonnegative trace term in~\eqref{cv:energy-qm-hartree-preparative}, but keeping the other terms as they are, and using $E_{a,b,N}^{\mathrm{H}} \leq 0$ with the upper bound in~\eqref{ineq:measure-deF} makes~\eqref{cv:energy-qm-hartree} reading now
	\[
		E_{a,b,N}^{\mathrm{H}} \geq E_{a,b,N}^{\mathrm{QM}} \geq E_{a,b,N}^{\mathrm{H}} - C\frac{N^{8\beta}}{L} - C \sqrt{\frac{\log L}{N}} \left( L+N^{4\beta}\right).
	\]
	Then, the upper bound $\beta < 1/16 + (\nu - 2\mu)/4$ is the largest range of $\beta$'s such that we can find $L\sim N^\delta$ such that the error terms are negligible compared to $E_{a,b,N}^{\mathrm{H}} \sim N^{\nu-2\mu}$. Note that in the above equation, we also took advantage of $\alpha<\beta$, which implies that all the error terms in~\eqref{lower-bound:truncated}--\eqref{lower-bound:deF} and in consecutive equations can be written w.l.o.g.\ with only $\beta$ powers of $N$ instead of a sum of $\alpha$ and $\beta$ powers. That is why we do not need any similar bound on $\alpha$.

	This concludes the proof of Theorem~\ref{thm:many-body-homogeneous}.
\end{proof}

To conclude this work, we note that Theorems~\ref{thm:collapse-nls}, \ref{thm:many-body}, \ref{thm:collapse-nls-homogeneous}, and~\ref{thm:many-body-homogeneous} leave out the case where $a>a_*$ is fixed and $b \searrow 0$. In such collapse regime, the blow-up phenomenon can be observed in the \NLS and Hartree theories, by using the nonlinearity of the corresponding functionals. It occurs in the sense that the energy tends to $-\infty$, as pointed out in Remarks~\ref{rem:collapse-nls-inhomogeneous} and~\ref{rem:collapse-nls-homogeneous}. For convenience, let us introduce the following homogeneous Hartree minimization problem
\begin{equation}\label{energy:hartree-homogeneous}
	G^{\mathrm{H}}_{a,b,N} := \inf\left\{\cG^{\mathrm{H}}_{a,b,N}(v):v\in H^1(\R^2), \norm{v}_{L^2}=1\right\},
\end{equation}
with the associated homogeneous Hartree functional
\begin{multline}\label{functional:hartree-homogeneous}
	\cG^{\mathrm{H}}_{a,b,N}(v) := \int_{\R^2}|\nabla v(x)|^2 \dix - \frac{a}{2}\iint_{\R^4} U_{N^\alpha}(x-y)|v(x)|^2|v(y)|^2 \dix \diy \\
	+ \frac{b}{6}\iiint_{\R^6} W_{N^\beta}(x-y, x-z)|v(x)|^2|v(y)|^2|v(z)|^2 \dix \diy \diz \,.
\end{multline}
As in Theorem~\ref{thm:collapse-hartree} and in the proof of Theorem~\ref{thm:many-body-homogeneous}, we prove the condensation and collapse of \emph{approximate} Hartree ground states both in the mean-field regime where $N\to+\infty$ and in the collapse regime where $a>a_*$ is fixed and $b_N \searrow 0$.

\begin{theorem}[Condensation and collapse of homogeneous Hartree approximate ground states]\label{thm:collapse-hartree-homogeneous}
	Let $0<\alpha<\beta$ and let $U$ and $W$ satisfy~\eqref{condition:potential-two-body}--\eqref{condition:potential-three-body-symmetry}.
	\begin{enumerate}[label=(\roman*)]
		\item\label{thm:collapse-homogeneous_hartree_item1} Let $a>a_*$ and $b>0$ be fixed. Let $\{v_N\}_N$ be a sequence of approximate ground states of~$G_{a,b,N}^{\mathrm{H}}$ defined in~\eqref{energy:hartree-homogeneous}. Then, there exists a homogeneous \NLS ground state~$v$ of~$G_{a,b}^{\NLS}$ defined in~\eqref{energy:nls-homogeneous} such that, up to translation and extraction of a subsequence,
		\begin{equation}\label{cv:hartree-nls-homogeneous-ground-state}
			\lim_{N\to+\infty}v_N = v
		\end{equation}
		strongly in~$H^1(\R^2)$. Furthermore,
		\begin{equation}\label{cv:hartree-nls-homogeneous-energy}
			\lim_{N\to+\infty}G_{a,b,N}^{\mathrm{H}} = G_{a,b}^{\NLS} \,.
		\end{equation}
		\item\label{thm:collapse-homogeneous_hartree_item2} Assume further that $|x|U(x) \in L^1(\R^2)$. Fix $a>a_*$ and assume $b_N = N^{-\eta}$ with $0 < \eta < 2\alpha$. Let $\{v_N\}_N$ be a sequence of approximate ground states of~$G_{a, b_N, N}^{\mathrm{H}}$. Then, up to translation and extraction of a subsequence,
		\begin{equation}\label{blowup:hartree-homogeneous-ground-state}
			\lim_{N\to+\infty} \sqrt{b_N} v_N \!\left(\sqrt{b_N} \, \cdot \right) = w
		\end{equation}
		strongly in~$H^1(\R^2)$, where $w$ is an optimizer of~$G_{a,1}^{\NLS}<0$ given by~\eqref{energy:nls-homogeneous}. Furthermore,
		\begin{equation}\label{blowup:hartree-homogeneous-energy}
			G_{a, b_N, N}^{\mathrm{H}} = G_{a,b_N}^{\NLS} (1+o(1)) = b_N^{-1} G_{a,1}^{\NLS} (1+o(1)) \,.
		\end{equation}
	\end{enumerate}
\end{theorem}

\begin{remark*}\leavevmode
	\begin{itemize}[leftmargin=*]
		\item By a simple scaling argument, we see that $G_{a,b}^{\NLS} = b^{-1} G_{a,1}^{\NLS}$ for all $a,b>0$.
		
		\item Similarly to the homogeneous \NLS case, since $G_{a, b_N, N}^{\mathrm{H}} \to -\infty$ ---see~\eqref{blowup:hartree-homogeneous-energy}---, the sequence of \emph{approximate} Hartree ground states $\{v_N\}_N$ is now naturally defined, instead of~\eqref{hartree:inhomogeneous-approximate-ground-state}, through
		\[
			G_{a, b_N, N}^{\mathrm{H}} \leq \cG_{a, b_N, N}^{\mathrm{H}}(v_N) \leq G_{a, b_N, N}^{\mathrm{H}} + o(1) \,.
		\]
		
		\item About~\emph{\ref{thm:collapse-homogeneous_hartree_item1}}, because the homogeneous Hartree energy can be split easily into many small pieces with same energy, a true homogeneous Hartree ground state may not exist and we have to consider \emph{approximate} ground states, in the sense of energy~\eqref{hartree:inhomogeneous-approximate-ground-state}. Their condensation~\eqref{cv:hartree-nls-homogeneous-ground-state} is obtained by using the semi-nonlinearity of the corresponding homogeneous Hartree functional.
		
		For comparison, still in the absence of an external potential, a true many\nobreakdash-body ground state does not exist due to the translation-invariance of the many\nobreakdash-body system. Although we could consider \emph{approximate} many\nobreakdash-body ground states (in the sense of energy), we do not hope for their condensation to occur (in the mean-field regime) due to their superposition in the many-body theory.
		
		\item About~\emph{\ref{thm:collapse-homogeneous_hartree_item2}}, that is in the collapse regime, where $a>a_*$ is fixed and $b_N \searrow 0$, it is worth noting that the blowup profiles~\eqref{blowup:hartree-homogeneous-ground-state} and~\eqref{blowup:hartree-homogeneous-energy} apply to both homogeneous and inhomogeneous Hartree problems, \eqref{energy:hartree-inhomogeneous} and~\eqref{energy:hartree-homogeneous}.
		
		\item Finally note that, in such collapse regime, the condensation of the true many\nobreakdash-body ground states of~\eqref{energy:quantum} would be obtained as soon as $G_{a,1}^{\NLS}$ admits a unique ground state, by the arguments in the proof of Theorems~\ref{thm:many-body} and~\ref{thm:collapse-hartree-homogeneous} (see the proof of Theorem~\ref{thm:many-body-homogeneous} for an explanation). Without this uniqueness of the blowup profile, one could proceed instead with a further argument by convex analysis (see, e.g.,~\cite{Nguyen-19b}). However, it requires the compactness of the sequence of many\nobreakdash-body ground states, which is hard to obtain, and not even obvious to be true, since the external term has no contribution to the leading order term in this collapse regime.
	\end{itemize}
\end{remark*}
\begin{proof}[Proof of Theorem~\ref{thm:collapse-hartree-homogeneous}]
	Let us start with~\emph{\ref{thm:collapse-homogeneous_hartree_item1}}. The proof of the energy upper bound in~\eqref{cv:hartree-nls-homogeneous-energy} is derived directly by the variational principle together with~\eqref{cv-ineq:hartree-nls-2body} and~\eqref{cv-ineq:hartree-nls-3body}. The matching energy lower bound is proved by justifying~\eqref{cv:hartree-nls-homogeneous-ground-state}. While the compactness of the sequence of \emph{true} Hartree ground states was obtained in Theorem~\ref{thm:collapse-hartree} for the inhomogeneous case, the result in the homogeneous case is more complicated and requires to consider \emph{approximate} homogeneous Hartree ground states, since a true homogeneous Hartree ground state does not exist, due to the translation-invariance of the homogeneous Hartree functional. The compactness of the sequence of approximate homogeneous Hartree ground states can be obtained from Lions' concentration-compactness method~\cite{Lions-84a}. Indeed, under the conditions that $a>a_*$, $b>0$ are fixed and $0<\alpha<\beta$, the boundedness of~$\{\norm{\nabla v_N}_{L^2}\}_N$ is obtained by the same arguments as in the proof of~\eqref{cv:hartree-nls-ground-state} ---note that the arguments therein do not depend on the external potential. Therefore, there exists $v\in H^1(\R^2)$ such that, up to translation and extraction of a subsequence, $v_N \to v$ weakly in~$H^1(\R^2)$ and almost everywhere in~$\R^2$. As usual, we now prove that this is actually strong convergence by showing $\norm{ v }_{L^2}=1$. Indeed, if $\norm{ v }_{L^2} = 0$ then we must have that $v_N \to 0$ strongly in~$L^p(\R^2)$ for all $2<p<\infty$ (see, e.g.,~\cite[Lemma 9]{LenLew-11}). Hence, by~\eqref{cv-ineq:hartree-nls-2body} and~\eqref{cv-ineq:hartree-nls-3body}, we have
	\[
		\lim_{N\to+\infty}G_{a,b,N}^{\mathrm{H}} \geq 0 \,.
	\]
	However, since $G_{a,b}^{\NLS} < 0$, the above contradicts the upper bound in~\eqref{cv:hartree-nls-homogeneous-energy}. Therefore, $v \not\equiv 0$. It remains to discard the possibility that $\lambda:=\norm{ v }_{L^2} \in (0,1)$. If this was the case, by~\eqref{cv-ineq:hartree-nls-2body},~\eqref{cv-ineq:hartree-nls-3body}, and the Brezis--Lieb refinement of Fatou's lemma (see, e.g.,~\cite[Theorem 1.9]{LieLos-01}), we would have
	\begin{align}\label{cv:hartree-nls-homogeneous-energy-decomposition}
		G_{a,b}^{\NLS} \geq \lim_{N\to+\infty}G_{a,b,N}^{\mathrm{H}} = \lim_{N\to+\infty}\cG_{a,b,N}^{\mathrm{H}}(v_N) & \geq \cG_{a,b}^{\NLS}(v) + \lim_{N\to+\infty}\cG_{a,b}^{\NLS}( v_N - v ) \nonumber \\
		& \geq G_{a,b}^{\NLS}(\lambda) + G_{a,b}^{\NLS}\!\left(\sqrt{1-\lambda^2}\right),
	\end{align}
	where we used the variational principle and the fact that $\norm{ v_N - v }_{L^2}^2 \to 1-\norm{ v }_{L^2}^2$, as $N\to+\infty$, for the last inequality, and where $G_{a,b}^{\NLS}(\lambda)$ denotes the ground state energy of the minimization problem of~\eqref{functional:nls-homogeneous} under the mass constrain $\norm{v}_{L^2}=\lambda\in (0,1)$. in particular, $G^{\NLS}_{a,b}=G^{\NLS}_{a,b}(1)$. It is worth noting that, by the same arguments as in Theorem~\ref{thm:collapse-nls-homogeneous}-\emph{\ref{existence:nls-homogeneous-gound-state}}, either $G_{a,b}^{\NLS}(\lambda)$ is nonnegative and admits no ground states or $G_{a,b}^{\NLS}(\lambda)$ is negative and admits a ground state. With this observation, we prove that~\eqref{cv:hartree-nls-homogeneous-energy-decomposition} cannot occur by considering the following possibilities on $G_{a,b}^{\NLS}(\lambda)$ and $G_{a,b}^{\NLS}\big(\sqrt{1-\lambda^2}\big)$.
	\begin{itemize}[leftmargin=*]
		\item If both are nonnegative, given that $G_{a,b}^{\NLS} < 0$, then
		\begin{equation}\label{cv:hartree-nls-homogeneous-binding-inequality}
			G_{a,b}^{\NLS}(\lambda) + G_{a,b}^{\NLS}\!\left(\sqrt{1-\lambda^2}\right) > G_{a,b}^{\NLS} \,,
		\end{equation}
		 contradicting~\eqref{cv:hartree-nls-homogeneous-energy-decomposition}.
		
		\item If one is nonnegative and one negative, w.l.o.g.\ we assume $G_{a,b}^{\NLS}(\lambda) < 0 \leq G_{a,b}^{\NLS}\big(\sqrt{1-\lambda^2}\big)$. Let $v_\lambda$ be a ground state of~$G_{a,b}^{\NLS}(\lambda)$. Since $\lambda \in (0,1)$ and $G_{a,b}^{\NLS} < 0$, we have
		\begin{equation}\label{cv:hartree-nls-homogeneous-binding-inequality-half-1}
			G_{a,b}^{\NLS}(\lambda) = \cG_{a,b}^{\NLS}(v_\lambda) \geq \lambda^2 \cG_{a,b}^{\NLS}(v_\lambda(\lambda\cdot)) \geq \lambda^2 G_{a,b}^{\NLS} > \lambda G_{a,b}^{\NLS} > G_{a,b}^{\NLS}
		\end{equation}
		contradicting also~\eqref{cv:hartree-nls-homogeneous-energy-decomposition} since $G_{a,b}^{\NLS}\big(\sqrt{1-\lambda^2}\big) \geq 0$.
		
		\item If both are negative, then we have, by the same argument as in~\eqref{cv:hartree-nls-homogeneous-binding-inequality-half-1}, that
		\begin{equation}\label{cv:hartree-nls-homogeneous-binding-inequality-half-2}
			G_{a,b}^{\NLS}\!\left(\sqrt{1-\lambda^2}\right) \geq \left(1-\lambda^2\right) G_{a,b}^{\NLS} > (1-\lambda) G_{a,b}^{\NLS}
		\end{equation}
		and~\eqref{cv:hartree-nls-homogeneous-energy-decomposition} is contradicted by~\eqref{cv:hartree-nls-homogeneous-binding-inequality-half-1} and~\eqref{cv:hartree-nls-homogeneous-binding-inequality-half-2}.
	\end{itemize}
	 Therefore, we must have $\norm{ v }_{L^2}=1$ and $v_N \to v$ strongly in~$L^p(\R^2)$ for all $2\leq p <\infty$. Moreover, again by the variational principle together with~\eqref{cv-ineq:hartree-nls-2body} and~\eqref{cv-ineq:hartree-nls-3body}, we have
	\[
		G_{a,b}^{\NLS} \geq \lim_{N\to+\infty}G_{a,b,N}^{\mathrm{H}} = \lim_{N\to+\infty}\cG_{a,b,N}^{\mathrm{H}}(v_N) \geq \cG_{a,b}^{\NLS}(v) \geq G_{a,b}^{\NLS} \,,
	\]
	yielding the energy lower bound in~\eqref{cv:hartree-nls-homogeneous-energy} as well as the $H^1(\R^2)$-convergence in~\eqref{cv:hartree-nls-homogeneous-ground-state}.
	
	\medskip
	
	Next, we prove~\emph{\ref{thm:collapse-homogeneous_hartree_item2}}. That is, the asymptotic behaviors of the homogeneous Hartree energy and its ground states in the regime $b=b_N \searrow 0$ and $a>a_*$ is fixed. Let $w$ be an optimizer of~$G_{a,1}^{\NLS}$. By the variational principle, the second inequality in~\eqref{cv-rate:hartree-nls-2body} and the nonnegativity in~\eqref{cv-ineq:hartree-nls-3body}, we have
	\begin{align*}
		G_{a, b_N, N}^{\mathrm{H}} &\leq \cG_{a,b_N}^{\mathrm{H}, N}(b_N^{-\frac{1}{2}}w(b_N^{-\frac{1}{2}}\cdot)) \\
		&\leq \cG_{a,b_N}^{\NLS}(b_N^{-\frac{1}{2}}w(b_N^{-\frac{1}{2}}\cdot)) + CN^{-\alpha} b_N^{-\frac{3}{2}} = \left(G_{a,1}^{\NLS} + CN^{-\alpha} b_N^{-\frac{1}{2}}\right) b_N^{-1} \,.
	\end{align*}
	The condition $0<\eta<2\alpha$ is therefore assumed to ensure that the error term $N^{-\alpha} b_N^{-\frac{1}{2}} = N^{\eta/2-\alpha}$ is negligible. This gives the upper bound in~\eqref{blowup:hartree-homogeneous-energy}. We prove the matching energy lower bound by proving~\eqref{blowup:hartree-homogeneous-ground-state}. Let $\{v_N\}_N$ be a sequence of approximate ground states of~$G_{a, b_N, N}^{\mathrm{H}}$ and $w_N := \sqrt{b_N} v_N \big(\sqrt{b_N} \cdot \big)$. Then, $\norm{w_N}_{L^2} = \norm{v_N}_{L^2} = 1$ and, by the nonnegativity in~\eqref{cv-ineq:hartree-nls-2body},
	\begin{multline}\label{blowup:hartree-homogeneous}
		b_N \cG_{a, b_N, N}^{\mathrm{H}}(v_N) \geq \frac{1}{6}\iiint_{\R^6} W_{N^\beta \sqrt{b_N}}(x-y, x-z)|w_N(x)|^2|w_N(y)|^2|w_N(z)|^2 \dix \diy \diz \\
		+ \int_{\R^2}\left(|\nabla w_N(x)|^2 - \frac{a}{2}|w_N(x)|^4 \right) \dix \,.
	\end{multline}
	Due to the assumptions $a>a_*$ fixed and $0<\alpha<\beta$, the boundedness of~$\{\norm{\nabla w_N}_{L^2}\}_N$ is obtained by the same arguments as in the proof of~\eqref{blowup:hartree-inhomogeneous-ground-state} ---note again that the arguments therein do not rely on the external potential. Therefore, there exists $w\in H^1(\R^2)$ such that, up to translation and extraction of a subsequence, $w_N \to w$ weakly in~$H^1(\R^2)$ and almost everywhere in~$\R^2$. We prove that the convergence is actually strong by showing $\norm{w}_{L^2}=1$. On the one hand, if $\norm{w}_{L^2} = 0$, then $w_N \to 0$ strongly in~$L^p(\R^2)$ for all $2<p<\infty$ (see, e.g.,~\cite[Lemma 9]{LenLew-11}). Hence, by~\eqref{cv-ineq:hartree-nls-3body} and~\eqref{blowup:hartree-homogeneous} where $N^\beta \sqrt{b_N} = N^{\beta-\eta/2} \to +\infty$, as $N\to+\infty$, due to the two assumptions $\beta > \alpha > \eta/2$, we have
	\[
		\lim_{N\to+\infty}b_N G_{a,b,N}^{\mathrm{H}} = \lim_{N\to+\infty}b_N \cG_{a, b_N, N}^{\mathrm{H}}(v_N) \geq 0 \,,
	\]
	contradicting the upper bound in~\eqref{blowup:hartree-homogeneous-energy} since $G_{a,1}^{\NLS} < 0$. Therefore, $w \not\equiv 0$. On the other hand, if $\lambda:=\norm{w}_{L^2} \in (0,1)$, then~\eqref{cv-ineq:hartree-nls-3body},~\eqref{blowup:hartree-homogeneous}, and the Brezis--Lieb refinement of Fatou's lemma (see, e.g.,~\cite[Theorem 1.9]{LieLos-01}) yield
	\begin{align*}
		G_{a,1}^{\NLS} \geq \lim_{N\to+\infty}b_N G_{a, b_N, N}^{\mathrm{H}} ={} & \lim_{N\to+\infty}b_N \cG_{a, b_N, N}^{\mathrm{H}}(v_N) \\
		\geq{} & \cG_{a,1}^{\NLS}(w) + \lim_{N\to+\infty}\cG_{a,1}^{\NLS}(w_N - w) \geq G_{a,1}^{\NLS}(\lambda) + G_{a,1}^{\NLS}\left(\sqrt{1-\lambda^2}\right),
	\end{align*}
	contradicting the binding inequality~\eqref{cv:hartree-nls-homogeneous-binding-inequality} for $b=1$. We have therefore proved $\norm{w}_{L^2}=1$ and consequently $w_N \to w$ strongly in~$L^p(\R^2)$ for all $2\leq p <\infty$. Moreover, again by the variational principle together with~\eqref{cv-ineq:hartree-nls-3body} and~\eqref{blowup:hartree-homogeneous}, we have
	\[
		G_{a,1}^{\NLS} \geq \lim_{N\to+\infty}b_N G_{a, b_N, N}^{\mathrm{H}} = \lim_{N\to+\infty}b_N \cG_{a, b_N, N}^{\mathrm{H}}(v_N) \geq \cG_{a,1}^{\NLS}(w) \geq G_{a,1}^{\NLS} \,,
	\]
	yielding the energy lower bound in~\eqref{blowup:hartree-homogeneous-energy} as well as the $H^1(\R^2)$-convergence in~\eqref{blowup:hartree-homogeneous-ground-state}.
\end{proof}

\end{document}